\documentclass{fancy-article}
\pdfoutput=1

\usepackage{tikz}
\usetikzlibrary{cd}

\usepackage[xcolor, hyperref, cleveref, notion, quotation, electronic]{knowledge}
\knowledgeconfigure{quotation, protect quotation={tikzcd}}
\knowledgeconfigure{diagnose line=true, diagnose bar=true}

\IfKnowledgePaperModeTF{
}{
    \knowledgestyle{intro notion}{color={Dark Ruby Red}, emphasize}
    \knowledgestyle{notion}{color={Dark Blue Sapphire}}
    \hypersetup{
        colorlinks=true,
        breaklinks=true,
        linkcolor={Dark Blue Sapphire}, 
        citecolor={Dark Blue Sapphire}, 
        filecolor={Dark Blue Sapphire}, 
        urlcolor={Dark Blue Sapphire},
    }
    \IfKnowledgeElectronicModeTF{
    }{
        \knowledgeconfigure{anchor point color={Dark Ruby Red}, anchor point shape=corner}
        \knowledgestyle{intro unknown}{color={Dark Gamboge}, emphasize}
        \knowledgestyle{intro unknown cont}{color={Dark Gamboge}, emphasize}
        \knowledgestyle{kl unknown}{color={Dark Gamboge}}
        \knowledgestyle{kl unknown cont}{color={Dark Gamboge}}
    }
}
\knowledge{notion}
 | parity games
 | parity game
 | Parity games

\knowledge{notion}
 | quasi-polynomial
 | quasi-po\-ly\-no\-mial

\knowledge{notion}
 | universal attractor decomposition algorithm
 | universal algorithm
 | universal algorithms

\knowledge{notion}
 | attractor decomposition
 | attractor decompositions

\knowledge{notion}
 | tree of attractor decomposition

\knowledge{notion}
 | small@game

\knowledge{notion}
 | small@tree

\knowledge{notion}
 | embeds@tree

\knowledge{notion}
 | embeds@dominion

\knowledge{notion}
 | interleaving

\knowledge{notion}
 | priority
 | priorities

\knowledge{notion}
 | subgame
 | subgames

\knowledge{notion}
 | strategy

\knowledge{notion}
 | trap
 | traps
 | trapping
 | Traps

\knowledge{notion}
 | dominion
 | dominia
 | Dominion

\knowledge{notion}
 | attractor
 | attractors
 | Attractors

\knowledge{notion}
 | adaptive empty-set early termination rule
 | empty-set early termination rule
 | adaptive
 | empty-set condition
 | empty-set rule

\knowledge{notion}
 | solution

\knowledge{notion}
 | fixpoint game
 | fixpoint games

\knowledge{notion}
 | flowery
 | flowery subgame
 | flowery subgames

\knowledge{notion}
 | limit-winning
 | Limit-winning

\knowledge{notion}
 | Concurrent parity games
 | concurrent parity games

\knowledge{notion}
 | set-based symbolic model of computation
 | symbolic algorithm
 | symbolic algorithms

\knowledge{notion}
 | symbolic set variables
 | symbolic space

\knowledge{notion}
 | primitive symbolic operations

\knowledge{notion}
 | embeddable decomposition theorem

\knowledge{notion}
 | dominion separation property
 | dominion separation theorem
 | dominion separation

\knowledge{notion}
 | McNaughton-Zielonka
 | McNaughton-Zielonka algorithm
 | McNaugh\-ton-Zielonka
 | McNaugh\-ton-Zielonka algorithm

\knowledge{notion}
 | Parys's
 | Parys
 | Parys's algorithm

\knowledge{notion}
 | Lehtinen-Schewe-Wojtczak algorithm
 | Lehtinen-Schewe-Wojtczak
 | Lehtinen-Schewe-Wojt\-czak

\knowledge{notion}
 | universal trees
 | universal
 | universality
 | uni\-ver\-sal

\knowledge{notion}
 | nested fixpoints
 | nested fixpoint equation
 | fixpoint equations
 | nested fixpoint equations
 | NFE
 | NFEs

\newrobustcmd\email[1]{\normalfont\texttt{#1}}

\newif\ifproofappendix
\newrobustcmd\labelwithproof[1]{%
\label{#1}%
\ifproofappendix%
\marginnote{\footnotesize{%
  First stated at page~\pageref{#1}.%
}}
\else%
\marginnote{\footnotesize{%
  See the proof of \Cref{#1} at page~\pageref{proof-#1}.%
}}%
\fi%
}
\newenvironment{proofappendix}[2]
  {
    \proofappendixtrue%
    #2*
    \proofappendixfalse%
    \begin{proof}[Proof of \Cref{#1}]
      \label{proof-#1}
  }
  { 
    \end{proof}
  }
\newrobustcmd\recall[1]{
  \proofappendixtrue%
    #1*
  \proofappendixfalse%
}

\newrobustcmd\interleaving{%
  \mathrel{%
    \withkl{\kl[\interleaving]}{%
      \cmdkl{\bowtie}
    }%
  }%
}
\knowledge{\interleaving}{notion}

\knowledgenewrobustcmd\complete[1]{\cmdkl{C_{#1}}}
\knowledgenewrobustcmd\parys[1]{\cmdkl{P_{#1}}}
\knowledgenewrobustcmd\succinct[1]{\cmdkl{S_{#1}}}

\knowledgenewrobustcmd{\flower}{\cmdkl{\mathcal{F}}}

\knowledgenewrobustcmd{\seq}[1]{\cmdkl{\bigl\langle}#1\cmdkl{\bigr\rangle}}

\knowledgenewrobustcmd{\Attr}{\cmdkl{\mathrm{Attr}}}

\knowledge{notion}
 | \solveEven
 
\knowledge{notion}
 | \solveOdd

\newcommand{\Even}{\mathrm{Even}}
\newcommand{\Odd}{\mathrm{Odd}}

\newcommand{\Gc}{\mathcal{G}}
\newcommand{\Hc}{\mathcal{H}}
\newcommand{\Ic}{\mathcal{I}}
\newcommand{\Jc}{\mathcal{J}}
\newcommand{\Kc}{\mathcal{K}}
\newcommand{\Pc}{\mathcal{P}}
\newcommand{\Tc}{\mathcal{T}}

\newcommand{\eset}[1]{\bigl\{\, #1 \,\bigr\}}

\graphicspath{{fig/}}

\usepackage{marginnote}
\usepackage[vlined]{algorithm2e}
\SetAlCapSty{textsc}
\SetAlCapNameSty{textit}

\usepackage[nottoc,numbib]{tocbibind} 

\usepackage{amsthm, thmtools, thm-restate}

\declaretheorem[name=Theorem]{theorem}
\declaretheorem[name=Proposition, sibling=theorem]{proposition}
\declaretheorem[name=Property, sibling=theorem]{property}

\declaretheorem[name=Corollary, sibling=theorem]{corollary}

\declaretheorem[name=Lemma, sibling=theorem]{lemma}

\title{Universal Algorithms for Parity Games and Nested Fixpoints\thanks{%
First version: January 2020. Full version of a paper accepted
in Lecture Notes in Computer Science, volume 13660. \sloppy
Emails: \email{marcin.jurdzinski[at]warwick.ac.uk},
\email{remi.morvan[at]ens-paris-saclay.fr} and \email{thejaswini.\allowbreak{}raghavan.1[at]warwick.ac.uk}.
}}
\author[1]{\href{https://orcid.org/0000-0003-3640-8481}{Marcin Jurdzi\'nski}}
\author[2,3]{\href{https://orcid.org/0000-0002-1418-3405}{Rémi Morvan}}
\author[1]{K. S. Thejaswini}
\affil[1]{Department of Computer Science, University of Warwick}
\affil[2]{École normale supérieure Paris-Saclay}
\affil[3]{LaBRI, Univ. Bordeaux, CNRS \& Bordeaux INP}
\date{August 2022}

\begin{document}

\maketitle

\begin{abstract}
  An "attractor decomposition" meta-algorithm for solving "parity games" is given that generalises the classic McNaughton-Zielonka algorithm and its recent "quasi-polynomial" variants due to Parys (2019), and to Lehtinen, Schewe, and Wojtczak (2019). The central concepts studied and exploited are "attractor decompositions" of "dominia" in "parity games" and the ordered trees that describe the inductive structure of "attractor decompositions". 

  The "universal algorithm" yields "McNaughton-Zielonka", "Parys", and "Lehtinen-Schewe-Wojtczak" algorithms as special cases when suitable "universal trees" are given to it as inputs. The main technical results provide a unified proof of correctness and structural insights into those algorithms. 
  
  Suitably adapting the "universal algorithm" for "parity games" to "fixpoint games" gives a "quasi-polynomial" time algorithm to compute "nested fixpoints" over finite complete lattices.
  
  The "universal algorithms" for "parity games" and "nested fixpoints" can be implemented symbolically.  It is shown how this can be done with $O(\lg d)$ "symbolic space" complexity, improving the $O(d \lg n)$ "symbolic space" complexity achieved by Chatterjee, Dvo\v{r}\'{a}k, Henzinger, and Svozil (2018) for "parity games", where $n$ is the number of vertices and $d$ is the number of distinct "priorities" in a "parity game".


  \emph{Acknowledgements.}
  The first and the third author had been supported by the EPSRC grant
  EP/P020992/1 (Solving Parity Games in Theory and Practice). The idea of the design of the universal algorithm has
  been discovered independently and later by Nathana\"el Fijalkow; we thank him for sharing his conjectures with us
  and exchanging ideas about adaptive tree-pruning rules. We also thank 
  our anonymous reviewers and Alexander Kozachinskiy for
  helpful comments on earlier drafts of the paper.

  \alert{This document contains internal hyperlinks, and is best read on an electronic device.}
\end{abstract}

\vfill

\section{Context}
\label{sec:context}

\subsection{Parity games and their significance}

"Parity games" play a fundamental role in automata theory, logic, and
their applications to verification~\cite{EJ91}, program
analysis~\cite{BKMP21,HS21}, and synthesis~\cite{GTW01,LMS20}.  
In particular, "parity games" are very intimately linked to the
problems of emptiness and complementation of non-deterministic
automata on trees~\cite{EJ91,Zie98}, model checking and satisfiability
checking of fixpoint logics~\cite{EJ91,EJS93,BW18}, fair
simulation relations~\cite{EWS05} or evaluation of
nested fixpoint expressions~\cite{HSC16,BKMP21,HS21}. 
It is a long-standing open problem whether "parity games" can be solved
in polynomial time~\cite{EJS93}. 

The impact of "parity games" goes well beyond their home turf of
automata theory, logic, and formal methods.
For example, an answer~\cite{Fri09} of a question posed originally for
"parity games"~\cite{VJ00} has strongly inspired major breakthroughs on
the computational complexity of fundamental algorithms in stochastic 
planning~\cite{Fea10} and linear optimization~\cite{Fri11,FHZ11}, and  
"parity games" provide the foundation for the theory of nested fixpoint
expressions used in program analysis~\cite{BKMP21,HS21} and
coalgebraic model checking~\cite{HSC16}.

\subsection{Related work}

The major breakthrough in the study of algorithms for solving "parity
games" occurred in 2017 when Calude, Jain, Khoussainov, Li, and
Stephan~\cite{CJKLS17} have discovered the first quasi-polynomial
algorithm. 
Three other---and seemingly distinctly different---techniques for
solving "parity games" in "quasi-po\-ly\-no\-mial" time have been proposed in
quick succession soon after: by Jurdzi\'nski and Lazi\'c~\cite{JL17}, 
Lehtinen~\cite{Leh18}, and Lehtinen, Parys, Schewe, and Wojtczak~\cite{LPSW22}. We would like to remark that \cite{LPSW22} is journal paper---describing
two "quasi-polynomial" time algorithms---combining a conference paper of Parys~\cite{Par19} and a 
preprint by Lehtinen, Schewe, and Wojtczak~\cite{LSW19}. To distinguish between the two algorithms, 
we refer to these versions as the algorithms by \AP""Parys"" and by 
\AP""Lehtinen-Schewe-Wojtczak"", respectively.

Czerwi\'nski, Daviaud, Fijalkow, Jurdzi\'nski, Lazi\'c, and
Pa\-rys~\cite{CDFJLP19} have also uncovered an underlying combinatorial
structure of "universal trees" as provably underlying the techniques of 
Calude et al., of Jurdzi\'nski and Lazi\'c, and of Lehtinen. 
Czerwi\'nski et al.\ have also established a quasi-poly\-no\-mial
lower bound for the size of smallest "universal trees", providing 
evidence that the techniques developed in those three papers may be
insufficient for leading to futher improvements in the complexity of
solving "parity games". 
The work of Lehtinen, Parys, Schewe, and Wojtczak~\cite{LPSW22},
who noted that the tree of recursive calls of their algorithms is "universal",
has not been obviously subject to the
quasi-polynomial barrier of Czerwi\'nski et al.~\cite{CDFJLP19},
making it a focus of current activity. 
Their algorithms are obtained by modifying
the classic "McNaughton-Zielonka algorithm"~\cite{McN93,Zie98}, which
has exponential running time in the worst case~\cite{Fri11r}, but
consistently outperforms most other algorithms in
practice~\cite{vDij18}. 

Using these "universal trees" as a crucial structure, there have also been further work to solve nested fixpoint expressions~\cite{HS21,ANP21} in quasi-polynomial time.

\subsection{Our contributions}

In this work we provide a meta-algorithm---the "universal
  attractor decomposition algorithm"---that generalizes
"McNaugh\-ton-Zielonka", "Parys's", and "Lehtinen-Schewe-Wojtczak"
algorithms.
There are multiple benefits of considering the "universal algorithm". 

Firstly, in contrast to "Parys's" and "Lehtinen-Schewe-Wojt\-czak"
algorithms, the "universal algorithm" has a very simple and transparent 
structure that minimally departs from the classic "McNaugh\-ton-Zielonka 
algorithm". 
Secondly, we observe that "Lehtinen-Schewe-Wojtczak algorithm", as well
as non-"adaptive" versions 
(see Sections~\ref{sec:universal-algorithm}
and~\ref{section:early-termination-heuristics}) 
of "McNaugh\-ton-Zielonka" and "Parys's" algorithms, all arise from the
"universal algorithm" by using specific classes of "universal trees",
strongly linking the theory of "universal trees" to the only class of
quasi-polynomial algorithms that had no established formal
relationship to "universal trees" so far. Moreover, since our algorithm can be modified to use any trees, they can also run on several classes of "universal trees" like the Strahler universal trees introduced in the 
work of Daviaud, Jurdzi\'nski and Thejaswini~\cite{DJT20}.

Thirdly, we further develop the theory of "dominia" and their "attractor
decompositions" in "parity games", initiated by Daviaud, Jurdzi\'nski,
and Lazi\'c~\cite{DJL18} and by Daviaud, Jur\-dzi\'n\-ski, and
Lehtinen~\cite{DJL19}, and we prove two new structural theorems
(the "embeddable decomposition theorem" and the "dominion separation
theorem") 
about ordered trees of "attractor decompositions".  

Fourthly, we use the structural theorems to provide a unified proof of
correctness of various "McNaughton-Zielonka"-style algorithms,
identifying very precise structural conditions on the trees of
recursive calls of the "universal algorithm" that result in it correctly
identifying the largest "dominia". 

Fifthly, we identify a structure of nested "fixpoint games", the "parity games" that arise naturally while solving fixpoint expressions which help us solve them in "quasi-polynomial" time using a modification of our "universal algorithm".

Finally, we observe that thanks to its simplicity, the "universal
algorithm" is particularly well-suited for solving "parity games" as well as "nested fixpoint equations"
efficiently in a symbolic model of computation, when large sizes of
input graphs prevent storing them explicitly in memory. 
Indeed, we argue that already a routine implementation of the
"universal algorithm" for "parity games" improves the state-of-the-art "symbolic space"
complexity of solving "parity games" in quasi-polynomial time from 
$O(d \lg n)$ to $O(d)$, but we also show that a more sophisticated 
symbolic data structure allows to further reduce the "symbolic space" of
the "universal algorithm" to~$O(\lg d)$. 
\section{Dominia and decompositions}
\label{sec:dominia}

\subsection{Strategies, traps, and dominia}

A \AP""parity game""~$\Gc$ consists of a finite directed graph 
$(V, E)$ together with a partition $(V_{\Even}, V_{\Odd})$ of the set of vertices~$V$,
and a function $\pi : V \to \eset{0, 1, \dots, d}$ that labels every
vertex~$v \in V$ with a non-negative integer~$\pi(v)$ called its
\AP""priority"".  
We say that a cycle is \emph{even} if the highest vertex "priority" on
the cycle is even; otherwise the cycle is \emph{odd}.
We say that a "parity game" is \AP$(n, d)$-""small@@game"" if it has at
most~$n$ vertices and all vertex "priorities" are at most~$d$. 

For a set~$S$ of vertices, we write $\Gc \cap S$ 
for the substructure of~$\Gc$ whose graph is the subgraph of~$(V, E)$
induced by the sets of vertices~$S$.
Sometimes, we also write $\Gc \setminus S$ to denote 
$\Gc \cap (V \setminus S)$. 
We assume throughout that every vertex has at least one outgoing edge, 
and we reserve the term \AP""subgame"" to substructures $\Gc \cap S$, 
such that every vertex in the subgraph of $(V, E)$ induced by~$S$ has
at least one outgoing edge. 
For a "subgame" $\Gc' = \Gc \cap S$, we sometimes write $V^{\Gc'}$ for
the set of vertices~$S$ that the "subgame"~$\Gc'$ is induced by. 
When convenient and if the risk of confusion is contained, we may
simply write $\Gc'$ instead of $V^{\Gc'}$.

A (positional) Even \AP""strategy"" is a set $\sigma \subseteq E$
of edges such that:
\begin{itemize}
\item
  for every $v \in V_{\Even}$, there is an edge $(v, u) \in \sigma$,
\item
  for every $v \in V_{\Odd}$, if $(v, u) \in E$ then $(v, u) \in \sigma$.
\end{itemize}
We sometimes call all the edges in such an Even "strategy"~$\sigma$ the 
\emph{strategy edges}, and the definition of an Even "strategy" requires
that every vertex in~$V_{\Even}$ has an outgoing strategy edge, and
every outgoing edge of a vertex in~$V_{\Odd}$ is a strategy edge. 

For a non-empty set of vertices $T$, we say that an Even
"strategy"~$\sigma$ \emph{traps Odd in $T$} if no strategy edge 
leaves~$T$, that is, $w \in T$ and $(w, u) \in \sigma$ imply $u \in T$. 
We say that a set of vertices~$T$ is a \AP""trap"" for Odd if there is
an Even "strategy" that traps Odd in~$T$. 

Observe that if~$T$ is a "trap" in a game~$\Gc$ then $\Gc \cap T$ is a
"subgame" of~$\Gc$.
For brevity, we sometimes say that a "subgame" $\Gc'$ is a "trap" if 
$\Gc' = \Gc \cap T$ and the set~$T$ is a "trap" in~$\Gc$. 
Moreover, the following simple \emph{``"trap" transitivity''} property 
holds:
if $T$ is a "trap" for Odd in game~$\Gc$ and $T'$ is a "trap" for Odd in
"subgame"~$\Gc \cap T$ then $T'$ is a "trap" for Odd in~$\Gc$.

For a set of vertices $D \subseteq V$, we say that an Even
"strategy"~$\sigma$ is an \emph{Even dominion strategy} on $D$ if:
  $\sigma$ traps Odd in~$D$ and
  every cycle in the subgraph $(D, \sigma)$ is even.
Finally, we say that a set~$D$ of vertices is an Even \AP""dominion""
if there is an Even dominion strategy on it.

Odd strategies, "trapping" Even, and Odd "dominia" are defined in an
analogous way by swapping the roles of the two players. 
It is an instructive exercise to prove the following two facts about
Even and Odd "dominia".

\begin{proposition}[Closure under union]
  If $D$ and $D'$ are Even (resp.\ Odd) "dominia" then $D \cup D'$ is
  also an Even (resp.\ Odd) "dominion". 
\end{proposition}

\begin{proposition}["Dominion" disjointness]
\label{prop:dominion-disjointness}
  If $D$ is an Even "dominion" and $D'$ is an Odd "dominion" then 
  $D \cap D' = \emptyset$. 
\end{proposition}

From closure under union it follows that in every "parity game", there
is the largest Even "dominion"~$W_{\Even}$ 
(which is the union of all Even "dominia") 
and the largest Odd "dominion"~$W_{\Odd}$
(which is the union of all Odd "dominia"), and from "dominion"
disjointness it follows that the two sets are disjoint. 
The positional determinacy theorem states that, remarkably, the
largest Even "dominion" and the largest Odd "dominion" form a partition of
the set of vertices.

\begin{theorem}[Positional determinacy~\cite{EJ91}]
\label{thm:positional-determinacy}
  Every vertex in a given "parity game" is either in the largest Even
  "dominion" or in the largest Odd "dominion". 
\end{theorem}

\subsection{Reachability strategies and attractors}

In a "parity game"~$\Gc$, for a target set of vertices~$B$
(``bullseye'') and a set of vertices~$A$ such that $B \subseteq A$, 
we say that an Even "strategy"~$\sigma$ is an 
\emph{Even reachability "strategy" to $B$ from~$A$} if every infinite 
path in the subgraph $(V, \sigma)$ that starts from a vertex in~$A$
contains at least one vertex in~$B$.

For every target set~$B$, there is the largest
(with respect to set inclusion) set from which there is an Even
reachability "strategy" to~$B$ in~$\Gc$;
we call this set the Even \AP""attractor"" to~$B$ in~$\Gc$ and denote
it by~$\AP\intro*\Attr_\Even^\Gc(B)$.
Odd reachability strategies and Odd "attractors" are
defined analogously.

We highlight the simple facts that if~$A$ is an "attractor" for a player
in~$\Gc$ then its complement $V \setminus A$ is a "trap" for her; and
that attractors are monotone operators: if $B' \subseteq B$ then the
"attractor" to~$B'$ is included in the "attractor" to~$B$.

\subsection{Attractor decompositions}

If $\Gc$ is a "parity game" in which all "priorities" do not exceed a
non-negative even number~$d$ then we say that 
\[
\Hc \: = \:
\seq{A, (S_1, \Hc_1, A_1), \dots, (S_k, \Hc_k, A_k)}
\]
is an Even \AP$d$-""attractor decomposition"" of~$\Gc$ if:
\begin{itemize}
\item
  $A$ is the Even "attractor" to the (possibly empty) set of vertices of
  "priority"~$d$ in~$\Gc$;
\end{itemize}
and setting $\Gc_1 = \Gc\setminus A$, for all $i = 1, 2, \dots, k$, we have: 
\begin{itemize}
\item
  $S_i$ is a 
  non-empty 
  "trap" for Odd in~$\Gc_i$ in which every vertex "priority" is at
  most~$d-2$;  
\item
  $\Hc_i$ is a $(d-2)$-"attractor decomposition" of 
  "subgame"~$\Gc \cap S_i$; 
\item
  $A_i$ is the Even "attractor" to $S_i$ in~$\Gc_i$;
\item
  $\Gc_{i+1} = \Gc_i \setminus A_i$;
\end{itemize}
and the game $\Gc_{k+1}$ is empty.
If $d = 0$ then we require that $k = 0$.

The following proposition states that if a "subgame" induced by a "trap"
for Odd has an Even "attractor decomposition" then the "trap" is an Even 
"dominion". 
Indeed, a routine proof argues that the union of all the reachability
strategies, implicit in the attractors listed in the decomposition, is
an Even "dominion" "strategy". 

\begin{proposition}
\label{prop:decomposition-dominion-even}
  If $d$ is even, $T$ is a "trap" for Odd in~$\Gc$, and there is
  an Even $d$-"attractor decomposition" of~$\Gc \cap T$, then $T$ is an
  Even "dominion" in~$\Gc$. 
\end{proposition}

By symmetry, the dual proposition holds for player Even,
assuming that $d$ is odd.

Attractor decompositions are
witnesses for the largest "dominia" and that the classic recursive
"McNaughton-Zielonka algorithm"
can be amended to produce such
witnesses. We provide the details of this claim in \Cref{sec:McNZ}.
Since "McNaughton-Zielonka algorithm" produces Even and Odd "attractor
decompositions", respectively, of "subgames" that are induced by sets of
vertices that are complements of each other, a by-product of its
analysis is a constructive proof of the positional determinacy theorem
(Theorem~\ref{thm:positional-determinacy}). 

\begin{theorem}
  "McNaughton-Zielonka algorithm" can be enhanced to produce both the
  largest Even and Odd "dominia", and an "attractor decomposition" of
  each.  
  Every vertex is in one of the two "dominia". 
\end{theorem}

\section{Universal trees and algorithms}

The running time of the "McNaughton-Zielonka algorithm" is, up to a
small polynomial factor, determined by the number of recursive calls it makes overall.
While numerous experiments indicate that the algorithm performs very
well on some classes of random games and on games arising from
applications in model checking, temporal logic synthesis, and
equivalence checking~\cite{vDij18}, 
it is also well known that there are families of
"parity games" on which "McNaughton-Zielonka algorithm" performs
exponentially many recursive calls~\cite{Fri11r}.

Parys~\cite{Par19} has devised an ingenious modification of
Mc\-Naugh\-ton-Zielonka algorithm that reduced the number of recursive
calls of the algorithm to \AP""quasi-polynomial"" number $n^{O(\lg n)}$ in
the worst case. 
Lehtinen, Schewe, and Wojt\-czak~\cite{LSW19} have slightly modified
"Parys's algorithm" in order to improve the running time from 
$n^{O(\lg n)}$ down to $d^{O(\lg n)}$ for $(n, d)$-"small@@game" "parity games". 
They have also made an informal observation that the tree of recursive
calls of their recursive procedure is "universal". 

In this paper, we argue that "McNaughton-Zielonka algorithm",
"Parys's algorithm", and "Lehtinen-Schewe-Wojt\-czak" algorithm are special 
cases of what we call a 
"universal attractor decomposition algorithm". 
The "universal algorithm" is parameterized by two ordered trees and we
prove a striking structural result that if those trees are capacious
enough to embed (in a formal sense explained later) ordered trees that
describe the ``shape'' of some "attractor decompositions" of the
largest Even and Odd "dominia" in a "parity game", then the "universal
algorithm" correctly computes the two "dominia". 
It follows that if the algorithm is run on two "universal trees" then it
is correct, and indeed we reproduce "McNaughton-Zielonka", "Parys's", and 
Lehtinen-Schewe-Wojtczak algorithms by running the "universal algorithm"
on specific classes of "universal trees". 
In particular, "Lehtinen-Schewe-Wojtczak algorithm" is obtained by using
the succinct "universal trees" of Jur\-dzi\'n\-ski and Lazi\'c~\cite{JL17}, 
whose size nearly matches the "quasi-polynomial" lower bound on the
size of "universal trees"~\cite{CDFJLP19}.

\subsection{Universal ordered trees}
\label{section:universal-trees}

\paragraph*{Ordered trees.}
Ordered trees are defined inductively;
an ordered tree is the trivial tree $\AP\intro*\seq{}$ or a sequence  
$\reintro*\seq{\Tc_1, \Tc_2, \dots, \Tc_k}$, where $\Tc_i$ is an ordered 
tree for every $i = 1, 2, \dots, k$. For an ordered tree~$\Tc$, we denote its
\emph{number of leaves} by $\mathrm{leaves}(\Tc)$ and its \emph{height} by $\mathrm{height}(\Tc)$, with the convention that the height of the trivial tree is zero. Moreover, we denote by $\seq{\Tc}^n$ the ordered tree $\seq{\Tc_1,\hdots,\Tc_i}$ where
$\Tc_i$ is a copy of $\Tc$ for each $i = 1, 2, \dots, n$.

\paragraph*{Trees of "attractor decompositions".} 

The definition of an "attractor decomposition" is inductive and we
define an ordered tree that reflects the hierarchical structure of an
"attractor decomposition".
If $d$ is even and 
\[
\Hc = \seq{A, (S_1, \Hc_1, A_1), \dots, (S_k, \Hc_k, A_k)}
\] 
is an Even \AP$d$-"attractor decomposition" then we define the 
""tree of attractor decomposition""~$\Hc$, denoted by $\Tc_{\Hc}$, 
to be the trivial ordered tree~$\seq{}$ if $k = 0$, and otherwise, to
be the ordered tree  
$\seq{\Tc_{\Hc_1}, \Tc_{\Hc_2}, \dots, \Tc_{\Hc_k}}$, where for every
$i = 1, 2, \dots, k$, tree $\Tc_{\Hc_i}$ is the "tree of attractor 
decomposition"~$\Hc_i$.  
Trees of Odd "attractor decompositions" are defined analogously.

Observe that the sets $S_1, S_2, \dots, S_k$ in an "attractor
decomposition" as above are non-empty and pairwise disjoint, which
implies that trees of "attractor decompositions" are small relative to
the number of vertices and the number of distinct "priorities" in a
"parity game". 
More precisely, we say that an ordered tree is \AP$(n, h)$-""small@@tree""
if its height is at most~$h$ and it has at most $n$ leaves.  
The following proposition can be proved by routine structural
induction.  

\begin{proposition}
\label{prop:tree-of-decomposition-is-small}
  If $\Hc$ is an "attractor decomposition" of an $(n, d)$-"small@@game" "parity
  game" 
  then its tree $\Tc_{\Hc}$ is $(n, \lceil d/2 \rceil)$-"small@@tree".
\end{proposition}

\paragraph*{Embedding ordered trees.}

Intuitively, an ordered tree "embeds@@tree" another if the latter 
can be obtained from the former by pruning some subtrees. 
More formally, every ordered tree \AP""embeds@@tree"" the trivial tree~$\seq{}$,
and $\seq{\Tc_1, \Tc_2, \dots, \Tc_k}$ \reintro(tree){embeds}
$\seq{\Tc_1', \Tc_2', \dots, \Tc_\ell'}$ if there are indices 
$i_1, i_2, \dots, i_{\ell}$, such that 
$1 \leq i_1 < i_2 < \cdots < i_{\ell} \leq k$
and for every $j = 1, 2, \dots, \ell$, we have that $\Tc_{i_j}$
"embeds@@tree"~$\Tc_j'$.

\paragraph*{Universal ordered trees.}
We say that an ordered tree is \AP$(n, h)$-""universal""~\cite{CDFJLP19} if
it "embeds@@tree" every $(n, h)$-"small@@tree" ordered tree.
The complete $n$-ary tree of height~$h$ can be defined by induction
on~$h$: 
if $h=0$ then $\AP\intro*\complete{n, 0}$ is the trivial tree~$\seq{}$, and if $h>0$
then $\reintro*\complete{n, h}$ is the ordered tree $\seq{\complete{n, h-1}}^n$. 
The tree $\complete{n, h}$ is obviously $(n, h)$-"uni\-ver\-sal" but its size
is exponential in~$h$. 

We define two further classes $\parys{n, h}$ and $\succinct{n, h}$ of 
$(n, h)$-uni\-ver\-sal trees, introduced respectively by Parys~\cite{Par19}
and by Jurdzi\'nski and Lazi\'c~\cite{JL17},
whose size is only "quasi-polynomial", and hence
they are significantly smaller than the complete $n$-ary trees of
height~$h$.  
Both classes are defined by induction on $n+h$.

\AP If $h = 0$ then both $\intro*\parys{n, h}$ and $\intro*\succinct{n, h}$ are defined to be the
trivial tree~$\seq{}$. 
If $h>0$ then $\reintro*\parys{n, h}$ is defined to be the ordered tree
\[
  \seq{\parys{\lfloor n/2 \rfloor, h-1}}^{\lfloor n/2 \rfloor} \cdot 
    \seq{\parys{n, h-1}} \cdot 
    \seq{\parys{\lfloor n/2 \rfloor, h-1}}^{\lfloor n/2 \rfloor}\,, 
\]
and $\reintro*\succinct{n, h}$ is defined to be the ordered tree
\[
  \succinct{\lfloor n/2 \rfloor, h} \cdot \seq{\succinct{n, h-1}} \cdot 
    \succinct{\lfloor n/2 \rfloor, h}\,. 
\]
The following proposition can easily be proven
by induction on $(n,h)$.

\begin{proposition}
  Ordered trees $\complete{n, h}$, $\parys{n, h}$ and $\succinct{n, h}$ are 
  $(n, h)$-"universal". 
\end{proposition}

A proof of "universality" of~$\succinct{n, h}$ is implicit in the work of
Jurdzi\'nski and Lazi\'c~\cite{JL17}, whose 
\emph{succinct multi-counters} are merely an alternative presentation 
of trees~$\succinct{n, h}$.
Parys~\cite{Par19} has shown that the number of leaves in
trees $\parys{n, h}$ is $n^{\lg n + O(1)}$ and Jurdzi\'nski and
Lazi\'c~\cite{JL17} have proved that the number of leaves in
trees $\succinct{n, h}$ is~$n^{\lg h + O(1)}$. 
Czerwi\'nski et al.~\cite{CDFJLP19} have established a
"quasi-polynomial" lower bound on the number of leaves in
$(n, h)$-"universal trees", which the size of~$\succinct{n, h}$ exceeds only by
a small polynomial factor.

\subsection{Universal algorithm}
\label{sec:universal-algorithm}

Every call of "McNaugh\-ton-Zielonka algorithm" 
(Algorithm~\ref{algo:Z-Solve}) repeats the main loop 
until the set returned by a recursive call is empty.
If the number of iterations for each value of~$d$ is large then the
overall number of recursive calls may be exponential in~$d$ in the
worst case, and that is indeed what happens for some families of hard
"parity games"~\cite{Fri11r}.

\begin{algorithm}[htb]
  \DontPrintSemicolon
  \SetKwFunction{solveEven}{$\kl[\solveEven]{\text{Univ}_{\Even}}$}
  \SetKwFunction{solveOdd}{$\kl[\solveOdd]{\text{Univ}_{\Odd}}$}
  \SetKwProg{fun}{procedure}{:}{}
  \AP\fun{\intro*\solveEven{$\Gc, d, \Tc^{\Even}, \Tc^{\Odd}$}}{
      let $\Tc^{\Odd} =
        \seq{\Tc^{\Odd}_1, \Tc^{\Odd}_2, \dots, \Tc^{\Odd}_k}$\;
      $\Gc_1 \leftarrow \Gc$\;
      \For{$i \leftarrow 1$ \KwTo $k$}{
        $D_i \leftarrow \pi^{-1}(d) \cap \Gc_i$\;
        $\Gc_i' \leftarrow \Gc_i \setminus \Attr_{\Even}^{\Gc_i}(D_i)$\;
        $U_i \leftarrow \solveOdd\bigl(\Gc_i', d-1,   
          \Tc^{\Even}, \Tc^{\Odd}_i\bigr)$\;
        $\Gc_{i+1} \leftarrow \Gc_i \setminus \Attr^{\Gc_i}_{\Odd}(U_i)$\;
      }
    \Return{$V^{\Gc_{k+1}}$}
  }
  \fun{\intro*\solveOdd{$\Gc, d, \Tc^{\Even}, \Tc^{\Odd}$}}{
    let $\Tc^{\Even} =
      \seq{\Tc^{\Even}_1, \Tc^{\Even}_2, \dots, \Tc^{\Even}_\ell}$\;
    $\Gc_1 \leftarrow \Gc$\;
    \For{$i \leftarrow 1$ \KwTo $\ell$}{
      $D_i \leftarrow \pi^{-1}(d) \cap \Gc_i$\;
      $\Gc_i' \leftarrow \Gc_i \setminus \Attr_{\Odd}^{\Gc_i}(D_i)$\;
      $U_i \leftarrow \solveEven\bigl(\Gc_i', d-1,
        \Tc^{\Even}_i, \Tc^{\Odd}\bigr)$\;
      $\Gc_{i+1} \leftarrow \Gc_i \setminus \Attr^{\Gc_i}_{\Even}(U_i)$\;
    }
    \Return{$V^{\Gc_{\ell+1}}$}
  }
  \caption{
    \label{algo:universal-parity}
    The \emph{""universal attractor decomposition algorithm""}.}
\end{algorithm}

In our "universal attractor decomposition algorithm" 
(Algorithm~\ref{algo:universal-parity}), 
every iteration of the main loop performs exactly the same actions as
in "McNaughton-Zielonka algorithm" 
(see Algorithm~\ref{algo:Z-Solve} and Figure~\ref{figure:McN-Z-iter}), 
but the algorithm uses a different mechanism to determine how many
iterations of the main loop are performed in each recursive call.   
In the mutually recursive procedures $\solveOdd$ and $\solveEven$,
this is determined by the numbers of children of the root in the input
trees $\Tc^{\Even}$ (the third argument) and~$\Tc^{\Odd}$ 
(the fourth argument), respectively.  
Note that the sole recursive call of $\solveOdd$ in the $i$-th
iteration of the main loop in a call of $\solveEven$ is given 
subtree $\Tc_i^{\Odd}$ as its fourth argument and, analogously, the
sole recursive call of $\solveEven$ in the $j$-th iteration of the
main loop in a call of $\solveOdd$ is given subtree $\Tc_j^{\Even}$ as
its third argument. 

In order to characterise the tree of recursive calls, let us define the \AP""interleaving"" operation on two ordered trees
inductively as follows:  
\AP$\seq{} \intro*\interleaving \Tc = \seq{}$ and 
$\seq{\Tc_1, \Tc_2, \dots, \Tc_k} \reintro*\interleaving \Tc = 
  \seq{\Tc \interleaving \Tc_1, \Tc \interleaving \Tc_2, \dots, \Tc \interleaving \Tc_k}$. 
Then the following simple proposition provides an explicit description
of the tree of recursive calls of our "universal algorithm". We state it only for the case where $d$ is even, but a similar proposition holds when $d$ is odd if trees $\Tc^{\Even}$ and $\Tc^{\Odd}$ are swapped in the statement.

\begin{proposition}
\label{prop:tree-of-recursive-calls-even}
  If $d$ is even 
  then the tree of recursive calls to the   procedure
  $\solveEven\bigl(\Gc, d, \Tc^{\Even}, \Tc^{\Odd}\bigr)$
  is the "interleaving" $\Tc^{\Odd} \interleaving \Tc^{\Even}$
  of trees~$\Tc^{\Odd}$ and~$\Tc^{\Even}$. 
\end{proposition}

The following elementary proposition helps estimate the size of an 
"interleaving" of two ordered trees and hence the running time of a call
of the "universal algorithm" that is given two ordered trees as inputs. 

\begin{proposition}
\label{prop:size-of-interleaving}

  If $\Tc$ and $\Tc'$ are ordered trees then:
  \begin{itemize}
  \item
    $\mathrm{height}(\Tc \interleaving \Tc') \leq 
      \mathrm{height}(\Tc) + \mathrm{height}(\Tc')$;
  \item
    $\mathrm{leaves}(\Tc \interleaving \Tc') \leq 
      \mathrm{leaves}(\Tc) \cdot \mathrm{leaves}(\Tc')$.
  \end{itemize}
\end{proposition}

In contrast to the "universal algorithm", the tree of recursive calls of
"McNaugh\-ton-Zielonka algorithm" is not pre-determined by a structure
separate from the game graph, such as the pair of trees~$\Tc^{\Even}$
and~$\Tc^{\Odd}$.
Instead, "McNaugh\-ton-Zielonka algorithm" determines the number of
iterations of its main loop adaptively, using the \AP""adaptive 
empty-set early termination rule"":
terminate the main loop as soon as $U_i = \emptyset$. 
We argue that if we add the "empty-set early termination rule" to 
the "universal algorithm" in which both trees $\Tc^{\Even}$ and
$\Tc^{\Odd}$ are the tree~$\complete{n, d/2}$ then its behaviour coincides
with "McNaughton-Zielonka algorithm".

\begin{proposition}
  The "universal algorithm" performs the same actions and produces the
  same output as "McNaugh\-ton-Zielonka algorithm" if it is run on an
  $(n, d)$-"small@@game" "parity game"
  and with both trees $\Tc^{\Even}$ and $\Tc^{\Odd}$ equal 
  to~$\complete{n, d/2}$, and if it uses the "adaptive empty-set early
  termination rule".   
\end{proposition}

The idea of using rules for implicitly pruning the tree of recursive
calls of a "McNaughton-Zielonka"-style algorithm that are significantly
different from the "adaptive empty-set early termination rule" is due to 
Parys~\cite{Par19}.
In this way, he has designed the first "McNaughton-Zielonka"-style
algorithm that works in "quasi-polynomial" time~$n^{O(\lg n)}$ in the
worst case, and Lehtinen, Schewe, and Wojtczak~\cite{LSW19} have
refined "Parys's algorithm", improving the worst-case running time down   
to~$n^{O(\lg d)}$.  
Both algorithms use two numerical arguments 
(one for $\Even$ and one for~$\Odd$) 
and ``halving tricks'' on those parameters, which results in pruning
the tree of recursive calls down to "quasi-polynomial" size in the worst
case. 
We note that our "universal algorithm" yields the algorithms of "Parys"
and of Lehtinen et al., respectively, if, when run on an 
$(n, d)$-"small@@game" "parity game" 
and if both trees $\Tc^{\Even}$ and $\Tc^{\Odd}$ set to be the 
$(n, d/2)$-"universal trees" $\parys{n, d/2}$ and~$\succinct{n, d/2}$, respectively.  

\begin{proposition}
\label{prop:universal-LSW}
  The "universal algorithm" performs the same actions and produces the
  same output as "Lehtinen-Schewe-Wojtczak algorithm" if it is run on an
  $(n, d)$-"small@@game" "parity game"
  with both trees $\Tc^{\Even}$ and $\Tc^{\Odd}$ equal to~$\succinct{n, d/2}$.   
\end{proposition}

The correspondence between the "universal algorithm" executed on 
$(n, d/2)$-"universal trees" $\parys{n, d/2}$ and "Parys's algorithm" is a bit 
more subtle.  
While both run in quasi-polynomial time in the worst case, the former
may perform more recursive calls than the latter.
The two coincide, however, if the the former is enhanced with a
simple adaptive tree-pruning rule similar to the "empty-set early
termination rule".  
The discussion of this and other adaptive tree-pruning 
rules will be better informed once we have dicussed
sufficient conditions for the correctness of our "universal algorithm". 
Therefore, we will return to elaborating the full meaning of the
following proposition in  
Section~\ref{section:early-termination-heuristics}. 

\begin{proposition}
\label{prop:universal-non-adaptive-Parys}
  The "universal algorithm" performs the same actions and produces the
  same output as a non-"adaptive" version of "Parys's algorithm" if it is
  run on an $(n, d)$-"small@@game" "parity games" 
  with both trees $\Tc^{\Even}$ and $\Tc^{\Odd}$ equal to~$\parys{n, d/2}$. 
\end{proposition}

\section{Correctness via structural theorems}
\label{sec:correctness-via-structural-theorems}

The classical proof of the correctness of "McNaughton-Zielonka algorithm"~\cite{AG11} essentially relies on claim that when one reaches the "empty-set condition",
then this proves that we've precisely computed the opponent's winning region.
The argument breaks down if the loop terminates before that
"empty-set condition" obtains. 
Instead, Parys~\cite{Par19} has developed a novel 
\emph{dominion separation technique} to prove correctness of his
algorithm and Lehtinen et al.~\cite{LSW19} use the same technique to
justify theirs.  

In this paper, we significantly generalize the "dominion separation"
technique of Parys, which allows us to intimately link the correctness
of our meta-algorithm to shapes (modelled as ordered trees) of
"attractor decompositions" of largest Even and Odd "dominia". 
We say that the "universal algorithm" is correct on a "parity game" if
$\solveEven$ returns the largest Even "dominion" and $\solveOdd$ returns
the largest Odd "dominion". 
We also say that an ordered tree~$\Tc$ ""embeds@@dominion"" a "dominion"~$D$ in
a "parity game"~$\Gc$ if it "embeds@@tree" the tree of some "attractor 
decomposition" of~$\Gc \cap D$. 
The main technical result we aim to prove in this section is the
sufficiency of the following condition for the "universal algorithm" to
be correct.

\begin{theorem}[Correctness of universal algorithm]
\label{thm:correctness-of-universal-algorithm}
  The "universal algorithm" is correct on a "parity game"~$\Gc$ if it is
  run on ordered trees~$\Tc^{\Even}$ and~$\Tc^{\Odd}$, such that
  $\Tc^{\Even}$ "embeds@@dominion" the largest Even "dominion" in~$\Gc$ and
  $\Tc^{\Odd}$ "embeds@@dominion" the largest Odd "dominion" in~$\Gc$. 
\end{theorem}

\subsection{Embeddable decomposition theorem}

Before we prove Theorem~\ref{thm:correctness-of-universal-algorithm}, in this section
we establish another technical result---the \AP""embeddable decomposition 
theorem""---that enables our generalization of Parys's "dominion
separation" technique.  
Its statement is intuitive: 
a "subgame" induced by a "trap" has a simpler "attractor decomposition"
structure than the whole game itself; 
its proof, however, seems to require some careful surgery.

\begin{restatable}[Embeddable decomposition]{theorem}{embeddabledecomposition}
  \labelwithproof{thm:embeddable-decomposition}
  If $T$ is a "trap" for Even in a "parity game"~$\Gc$ and 
  $\Gc' = \Gc \cap T$ is the "subgame" induced by~$T$, 
  then for every Even "attractor decomposition"~$\Hc$ of~$\Gc$, 
  there is an Even "attractor decomposition"~$\Hc'$ of $\Gc'$, such
  that $\Tc_{\Hc}$ "embeds@@tree" $\Tc_{\Hc'}$. 
\end{restatable}

In order to streamline the proof of the "embeddable decomposition
theorem", we state the following two propositions, which synthesize or
generalize some of the arguments that were also used by
Lehtinen, Parys, Schewe and Wojtczak~\cite{LPSW22}.   
Proofs are included in the Appendix.

\begin{proposition}
  \label{prop:dominion-cap-trap}
  Suppose that $R$ is a "trap" for Even in game~$\Gc$.
  Then if $T$ is a trap for Odd in~$\Gc$ then $T \cap R$ is a trap for
  Odd in "subgame"~$\Gc \cap R$, and if $T$ is an Even "dominion" in~$\Gc$
  then $T \cap R$ is an Even dominion in~$\Gc \cap R$.  
\end{proposition}

The other proposition is illustrated in
Figure~\ref{figure:complement-of-attractor-in-trap}.
Its statement is more complex than that of the first proposition. 
The statement and the proof describe the relationship between
the Even "attractor" of a set~$B$ of vertices in a game~$\Gc$ and the 
Even attractor of the set $B \cap T$ in "subgame"~$\Gc \cap T$, where
$T$ is a "trap" for Even in~$\Gc$.

\begin{figure}
\centering
\begin{minipage}{.48\textwidth}
  \centering
  \makebox[\textwidth][c]{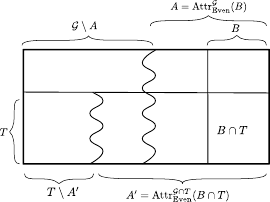}
  \caption{"Traps" and "attractors" in
    Proposition~\ref{prop:attractor-inside-trap}.} 
  \label{figure:complement-of-attractor-in-trap}
\end{minipage}%
\hfill%
\begin{minipage}{.48\textwidth}
  \centering
  \makebox[\textwidth][c]{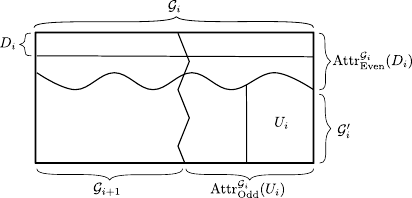}
  \caption{"Attractors" and "subgames" in one iteration of the loop in
    "attractor decomposition" algorithms.} 
  \label{figure:McN-Z-iter}
\end{minipage}
\end{figure}

\begin{proposition}
  \label{prop:attractor-inside-trap}
  Let $B \subseteq V^{\Gc}$ and let $T$ be a "trap" for Even in
  game~$\Gc$. 
  Define $A = \Attr_\Even^\Gc(B)$
  and $A' = \Attr_\Even^{\Gc \cap T}(B \cap T)$.
  Then $T \setminus A'$ is a "trap" for Even in
  "subgame"~$\Gc \setminus A$.
\end{proposition}

We prove the "embeddable decomposition theorem" by
induction on the number of leaves of the "tree of attractor
decomposition"~$\Hc$.
Note that our definition of an "attractor decomposition"
allows for $S_i$ to be any non-empty "trap" for Odd in~$\Gc_i$ in which every vertex "priority" is at
most~$d-2$, whereas Daviaud, Jur\-dzi\'n\-ski, and Lehtinen's  definition~\cite{DJL19}
ask for $S_i$ to be the \emph{maximal} "trap" for Odd satisfying the aforementioned property.
Relaxing the definition of "attractor decompositions" is crucial for Proposition~\ref{prop:attractor-inside-trap}
to hold.

\subsection{Dominion separation theorem}
\label{section:dominion-separation-theorem}

The simple dominion disjointness property
(Proposition~\ref{prop:dominion-disjointness}) 
states that every Even "dominion" is disjoint from every Odd "dominion". 
For two sets $A$ and $B$, we say that another set~$X$ 
\emph{separates $A$ from~$B$} if $A \subseteq X$ and 
$X \cap B = \emptyset$. 
In this section we establish a very general ""dominion separation
property"" for "subgames" that occur in iterations of the "universal
algorithm".
This allows us to prove one of the main technical results of this
paper (Theorem~\ref{thm:correctness-of-universal-algorithm})
that describes a detailed structural sufficient condition for the
correctness of the "universal algorithm". 

\begin{restatable}[Dominion separation]{theorem}{dominionseparation}
  \labelwithproof{thm:dominion-separation}
  Let $\Gc$ be an $(n, d)$-"small@@game" "parity game" and let 
  $\Tc^{\Even} = \seq{\Tc^{\Even}_1, \dots, \Tc^{\Even}_{\ell}}$
  and
  $\Tc^{\Odd} = \seq{\Tc^{\Odd}_1, \dots, \Tc^{\Odd}_k}$
  be trees of height at most $\lceil d/2 \rceil$ and
  $\lfloor d/2 \rfloor$, respectively.
  
  \begin{enumerate}[label=(\alph*)]
  \item
    \label{item:dominion-separation-a}
    If $d$ is even and $\Gc_1, \dots, \Gc_{k+1}$ are the games that are
    computed in the successive iterations of the loop in the call 
    $\solveEven\bigl(\Gc, d, \Tc^{\Even}, \Tc^{\Odd}\bigr)$, 
    then for every $i = 0, 1, \dots, k$, we have that $\Gc_{i+1}$
    separates every Even "dominion" in~$\Gc$ that tree $\Tc^{\Even}$
    "embeds@@dominion" from every Odd "dominion" in~$\Gc$ that tree  
    $\seq{\Tc_1^{\Odd}, \hdots, \Tc_i^{\Odd}}$
    "embeds@@dominion".   

  \item
    \label{item:dominion-separation-b}
    If $d$ is odd and $\Gc_1, \dots, \Gc_{\ell+1}$ are the games that
    are computed in the successive iterations of the loop in the call 
    $\solveOdd\bigl(\Gc, d, \Tc^{\Even}, \Tc^{\Odd}\bigr)$, 
    then for every $i = 0, 1, \dots, \ell$, we have that $\Gc_{i+1}$
    separates every Odd "dominion" in~$\Gc$ that tree $\Tc^{\Odd}$
    "embeds@@dominion" from every Even "dominion" in~$\Gc$ that tree 
    $\seq{\Tc_1^{\Even}, \hdots, \Tc_i^{\Even}}$ "embeds@@dominion".   
  \end{enumerate}
\end{restatable}


\subsection{Correctness and complexity}

The "dominion separation theorem"
(Theorem~\ref{thm:dominion-separation}) 
allows us to conclude the proof of the main "universal algorithm"
correctness theorem
(Theorem~\ref{thm:correctness-of-universal-algorithm}).   
Indeed, if trees $\Tc^{\Even}$ and $\Tc^{\Odd}$ satisfy the conditions
of Theorem~\ref{thm:correctness-of-universal-algorithm} then, by the
"dominion separation theorem", the set returned by the call 
$\solveEven\bigl(\Gc, d, \Tc^{\Even}, \Tc^{\Odd}\bigr)$ 
separates the largest Even "dominion" from the largest Odd "dominion", and
hence---by the positional determinacy theorem
(Theorem~\ref{thm:positional-determinacy})---it is the largest Even
dominion. 
The argument for procedure~$\solveOdd$ is analogous. 

We note that the "universal algorithm" correctness theorem, together
with Propositions~\ref{prop:universal-non-adaptive-Parys}
and~\ref{prop:universal-LSW}, imply correctness of the non-"adaptive"
version of "Parys's algorithm"~\cite{Par19} and of
"Lehtinen-Schewe-Wojtczak algorithm"~\cite{LSW19}, because trees of
"attractor decompositions" are $(n, d/2)$-"small@@tree" 
(Proposition~\ref{prop:tree-of-decomposition-is-small}) and 
trees~$\parys{n, d/2}$ and~$\succinct{n, d/2}$ are $(n, d/2)$-"universal".

The following fact, an alternative restatement of the conclusion of
Lehtinen et al.~\cite{LSW19}, is a simple corollary of the precise
asymptotic upper bounds on the size of the "universal 
trees"~$\succinct{n, d/2}$ established by Jurdzi\'nski and
Lazi\'c~\cite{JL17}, and of Propositions~\ref{prop:universal-LSW}, 
\ref{prop:tree-of-recursive-calls-even}, 
and~\ref{prop:size-of-interleaving}. 

\begin{proposition}[Complexity]
\label{prop:complexity-of-LSW}
  The "universal algorithm" that uses "universal trees" $\succinct{n, d/2}$
  (aka.\ Lehtinen-Sche\-we-Wojt\-czak algorithm) solves $(n, d)$-"small@@game"
  "parity games" in polynomial time if $d = O(\log n)$, and in time  
  $n^{2\lg (d/{\lg n}) + O(1)}$ if $d = \omega(\log n)$. 
\end{proposition}

\subsection{Acceleration by tree pruning}
\label{section:early-termination-heuristics}

As we have discussed in Section~\ref{sec:universal-algorithm},
Parys~\cite{Par19} has achieved a breakthrough of developing the first
"quasi-polynomial" "McNaughton-Zielonka"-style algorithm for "parity games"
by pruning the tree of recursive calls down to "quasi-polynomial" size.  
Proposition~\ref{prop:universal-non-adaptive-Parys} clarifies that
Parys's scheme can be reproduced by letting the "universal algorithm"
run on "universal trees"~$\parys{n, d/2}$, but as it also mentions, just
doing so results in a ``non-"adaptive"'' version of "Parys's algorithm".   
What is the ``"adaptive"'' version actually proposed by "Parys"? 

Recall that the root of tree $\parys{n, h}$ has $n+1$ children, the
first $n/2$ and the last~$n/2$ children are the roots of copies of  
tree~$\parys{n/2, h-1}$, and the middle child is the root of a copy
of tree~$\parys{n, h-1}$. 
The "adaptive" version of "Parys's algorithm" also uses another
tree-prunning rule, which is adaptive and a slight generalization of
the "empty-set rule":
whenever the algorithm is processing the block of the first $n/2$
children of the root or the last $n/2$ children of the root, if one of
the recursive calls in this block returns an empty set then the rest
of the block is omitted. 

We expect that our structural results 
(such as Theorems~\ref{thm:correctness-of-universal-algorithm}
and~\ref{thm:dominion-separation})
will provide insights to inspire development and proving correctness
of further and more sophisticated adaptive tree-pruning rules, but we
leave it to future work.   
This may be critical for making quasi-polynomial versions of
"McNaughton-Zielonka" competitive in practice with its basic version
that is exponential in the worst case, but remains very hard to beat
in practice~\cite{vDij18,LPSW22}. 

\section{Computing nested fixpoints}
\label{sec:nested-fixpoints}

Computing fixpoints is fundamental in the study of computer science. Solving "nested fixpoint equations" ("NFEs") over finite lattices are known to be computationally equivalent to solving "parity games"~\cite{BKMP21}, however, most of the reductions involve an exponential increase in the size of the resulting "parity game". The satisfiability problem of the coalgebraic $\mu$-calculus has also been 
reduced to the same~\cite{HS19}. A corollary of Calude et. al's breakthrough result was that specific kinds of "fixpoint equations" could be solved in "quasi-polynomial" time. Following this progress, there were several algorithms targeted at solving more general "fixpoint equations" by using universal graphs\cite{HS21} and "universal trees"~\cite{ANP21}. Hausman and Schr\"{o}der gave a quasi-polynomial algorithm to solve "NFEs" using progress measures on universal graphs whereas Arnold, Niwinski and Parys solved "NFEs" using the key result on decompositions of "dominia" similar to an earlier version of this paper. 
Here, we provide a slightly different way of solving "nested fixpoints" by converting the equation to an exponentially sized "fixpoint game", as in~\cite{HS21} but using our "universal attractor decomposition algorithm", parameterised by two trees, as in~\cite{ANP21}. The algorithm proposed by Arnold, Niwinski and Parys is similar to ours, in the sense that
both algorithms use a pair of trees to guide the computation of a subset of a complete lattice and
of a set of vertices in a parity game in our case, respectively. Since we can
describe the set of winning vertices for some player in a parity game with a formula
whose length is linear in the number of distinct "priorities": $d$, the algorithm of \cite{ANP21}
can be seen as a generalisation of Algorithm \ref{algo:universal-parity}. On the other hand, we
explain in this section how to, given a "nested fixpoint equation", run the latter algorithm
on a "parity game"---called a "fixpoint game"---, which has an exponential size compared to the size
of the "NFE", without having an exponential blowup. We thus obtain an algorithm to compute nested fixpoint equations in quasi-polynomial time. In this sense, we argue that the algorithm of Arnold,
Niwinski and Parys is equivalent to ours. However, it should be noted that \cite{ANP21} provides
an asymmetrical version of their algorithm---using a technique of Seidl~\cite{S96}---which is quadratically faster, in the worse case, than Algorithm~\ref{algo:universal-parity}.
In section \ref{section:symbolic}, a detailed description 
of how to implement symbolically both variants of the "universal algorithm", for parity games
and nested fixpoint, which require logarithmically less "symbolic space" than
Chatterjee, Dvo\v{r}\'ak, Henzinger and Svozil quasi-polynomial "symbolic algorithm" \cite{CDHS18} is provided.

We argue that we can directly apply our "universal attractor decomposition algorithm" on these exponential sized "fixpoint games" with the help of a carefully designed data structure, which ensures that we can in fact compute fixpoints using our algorithm in time proportional to $|\Tc_\Odd|\cdot|\Tc_\Even|$. 
\subsection{Nested Fixpoint Equations}
In this subsection, we will define "nested fixpoint equations" over the powerset lattice.
Consider a finite set of elements $U$ and its powerset lattice $\Pc(U)$. Let $f$ be a monotone function (component wise) from $\Pc(U)^d$ to  $\Pc(U)^d$.
The function $f$ can be expressed as a tuple $(f_1, \hdots, f_d)$ of functions
from $\Pc(U)^d$ to $\Pc(U)$, where $f_i$ is the projection
of $f$ to the $i$-th component.

Since there is a natural bijection from $d$ tuples of subsets of $U$ to subsets of ${(U\times [d])}$, we instead denote $f$ as a function from $\Pc(U\times [d])$ to $\Pc(U\times [d])$.

A \AP""nested fixpoint equation"" is a system of $d$ fixpoint equations of the form:
\begin{equation}\label{eqn:nfe}
    X_i \:=_{\eta_i} f_i(X_1,\ldots, X_{d})\tag{$*$}
\end{equation}
for $i$ ranging from $1,\dots d$ and where $\eta_i=\nu$ if $i$ is even, and $\eta_i = \mu$ otherwise. We refer to a system such as ~(\ref{eqn:nfe}) as a "nested fixpoint equation" and refer to it with the short hand: $X =_\eta f(X)$. One could consider a more general form of "fixpoint equations" where $\eta_i \in\{\mu,\nu\}$, but for simplicity of presentation, we restrict ourselves to the above. 

The \AP""solution"" of a system of $d$ "fixpoint equations" as the one defined by~(\ref{eqn:nfe}), is a subset of $U\times [d]$, defined recursively as follows. We say that the solution of the empty set of equations is the empty tuple. For a system of one or more fixpoint equations, we define a function $f^{d-1}$ from subsets of $U$ to subsets of ${\bigl(U\times [d-1]\bigr)}$. This function $f^{d-1}$ takes as input $Y_d$, a subset of $U$, and uses this input to fix $X_d\:= Y_d$ in the system of equations and the solution obtained to the system of $d-1$ equations by fixing $X_d$ to be $Y_d$ is the output of $f^{d-1}$. 
We finally say the solution of the system of equations is $\bigl(f^{d-1}(Y_d), Y_d\bigr)$, where $Y_d = \eta_d\bigl(\lambda X_d.f_d(f^{d-1}(X_d),X_d)\bigr)$. 

\subsection{Fixpoint Games}
Let us now define an equivalent "parity game" $\Gc_f$, called a \AP""fixpoint game"". Solving the parity game $\Gc_f$ correlates to finding the "solution" of the system of "nested fixpoint equation" defined by $X =_\eta f(X)$~\cite{BKMP21,HS21}.

Here, $\Gc_f = (V_f, E_f)$ with the "priority" function $\pi_f$, where $V_f$ consists of the disjoint union $\bigl(U\times [d]\bigr)\cup \{v_A\mid A\subseteq U \times [d]\}$. The vertices corresponding to elements of the set $\bigl(U\times[d]\bigr)$ belong to Even and the ones corresponding to subsets of the same set belong to Odd. The priority function $\pi_f$ assigns Even's vertices $(u,i)$ to $i$, and vertices Odd's vertices to priority $0$.
The edges from a vertex $(u,i)$ belonging to Even in $\Gc_f$ lead to the set of Odd vertices $\{v_A\mid (u,i)\in f(A)\}$ and edges from a vertex $v_A$, belonging to Odd lead to the set of Even vertices $\{(u,i)\mid (u,i)\in A\}$.

Finding if $(u,i)$ is in the "solution" of a "nested fixpoint equation" $X =_\eta f(X)$ is known to be equivalent to solving the corresponding "fixpoint game" $\Gc_f$ of the equation from the even vertex $(u,i)$, as shown in Theorem 4.8 of~\cite{BKMP21}.

\subsection{Solving Fixpoint Games}
We provide a way to solve a "fixpoint game" with the help of the "universal attractor decomposition algorithm" in Section~\ref{sec:universal-algorithm}.

We define a specific kind of "subgames" that we call \AP""flowery subgames"" and show that they
are pertinent to solving "fixpoint games" using the "universal attractor decomposition algorithm". 
Given two subsets $\emptyset \subsetneq Y \subseteq X \subseteq U\times [d]$, we define the flowery subgame
on $(X,Y)$, denoted by $\flower(X,Y)$, to be the "subgame" of $\Gc_f$ whose set of vertices consists of all Odd vertices $v_A$ which is a subset of $X$ intersecting non-trivially with $Y$, resembling the petal of a flower along with all vertices of Even belonging to $Y$, resembling the core of a flower. More formally, we define
\[
    \AP\intro*\flower(X,Y) \:= Y\uplus\{v_A\mid A\subseteq X\text{ and } A\cap Y \neq \emptyset\}.
\]

In the game $\Gc_f$, on removing vertices that have no outgoing edges along with the respective "attractors" to these sets of vertices, i.e., Odd "attractors" to Even vertices with no outgoing edges and vice versa, we get a "flowery subgame". 
Moreover, the following lemma reassures us that all significant operations performed by Algorithm~\ref{algo:universal-parity} on flowery subgames, 
results in flowery subgames.
\begin{restatable}[Floweriness]{lemma}{flowerylemma}
    \labelwithproof{lemma:flowery}
    If \solveEven (resp., \solveOdd) is run on a "flowery subgame", for all iterations in the for-loop, "subgame" $\Gc_i$  is also "flowery". In particular, $\Gc_{k+1}$, which is the subgame returned, is flowery.
\end{restatable}
The "attractor" to a set of vertices during a run of the algorithm can be computed by at most $d|U|$ many computation of $f$ on subsets of $U\times [d]$.  We can therefore solve nested "fixpoint games" in quasi-polynomial time using the "universal attractor decomposition algorithm", by only keeping track of the sets $X$ and $Y$ representing each subgame, as stated below.

\begin{theorem}\label{thm:nfe-complexity}
The modified "universal algorithm" that computes "nested fixpoint equations" on trees $\Tc_{\Odd}$ and $\Tc_{\Even}$ makes $|\Tc_{\Odd}|\cdot|\Tc_{\Even}|$ many recursive calls. Each recursive call makes at most $2d|U|$ many function evaluations of $f$. 
\end{theorem}
\subsection{Concurrent Parity games}

\AP""Concurrent parity games"" have been well studied before. We consider the two player version as studied by Chatterjee, Alfredo and Henzinger in~\cite{CAH11}. These games are played among two players---Even and Odd, but instead of partitioning the vertices among the two players, they take simultaneous actions at each vertex and the token moves to a neighbour depending on the actions of both players. One might also consider a stochastic version where the simultaneous actions are decided by a pre-decided probability distribution. Both the players are allowed to use a randomised "strategy", i.e., a strategy where the next action is proposed with the help of a probability distribution. \AP A state is called ""limit-winning"" for Even (resp. Odd) if Even (resp. Odd) has a strategy to win from that state with probability arbitrarily close to 1. 
The decision question we have at hand, is to determine if a state is a "limit-winning" state for a given input player. 
"Concurrent parity games" vary from original "parity games" in that, a player might need both infinite memory and randomisation to win these games. 
We refer the readers to the work of Chatterjee, Alfaro, and  Henzinger~\cite{CAH11} for a rigorous definition of the above games along with examples for the claims above. In their paper, they show that solving "concurrent parity games" is in $\text{NP}\cap \text{co-NP}$ as a corollary of the following theorem.

\begin{theorem}[{\cite[Theorem~5, Lemma~29 and Lemma~30]{CAH11}}]\label{thm:Concurrent-parity}
"Limit-winning" in a concurrent parity game can be expressed as an "NFE" over the powerset lattice of the set of edges with alternation depth at most $2d$ for a function, whose evaluation involves solving another "NFE" also with depth 
at most $2d$. 
\end{theorem}
An easy corollary from Theorem~\ref{thm:nfe-complexity} along with Theorem~\ref{thm:Concurrent-parity}, we have the following.
\begin{corollary}
"Limit-winning" in "concurrent parity games" can be solved in "quasi-polynomial" time.
\end{corollary}
\section{Symbolic algorithms}
\label{section:symbolic}

"Parity games" that arise in applications, for example
from the automata-theoretic model checking approaches to verification
and automated synthesis, often suffer from the 
\emph{state-space explosion problem}:
the sizes of models are exponential (or worse) in the sizes of natural 
descriptions of the modelled objects, and hence the models obtained
may be too large to store them explicitly in memory. 
One method of overcoming this problem that has been successful in the
practice of algorithmic formal methods is to represent the models 
symbolically rather than explicitly, and to develop algorithms for
solving the  models that work directly on such succinct symbolic
representations~\cite{BCMDH92}.

We adopt the \AP""set-based symbolic model of computation"" that was 
already considered for "parity games" by Chatterjee, Dvo\v{r}\'{a}k,
Henzinger, and Svozil~\cite{CDHS18}. 
In this model, any standard computational operations on any standard
data structures are allowed, but there are also the following symbolic 
resources available:  
  \AP""symbolic set variables"" can be used to store sets of vertices
  in the graph of a parity game;
  basic set-theoretic operations 
  on "symbolic set variables" are available as 
  \AP""primitive symbolic operations"";
  the \emph{controllable predecessors operations} are available as 
  "primitive symbolic operations":
  the Even (resp.\ Odd) controllable predecessor, when
  applied to a symbolic set variable~$X$, returns the set of vertices
  from which Even (resp.\ Odd) can force to move into the set~$X$, by
  taking just one outgoing edge.  
Since "symbolic set variables" can represent possibly very large and
complex objects, they should be treated as a costly resource.   

Chatterjee et al.~\cite{CDHS18} have given a symbolic set-based
algorithm that on $(n, d)$-"small@@game" "parity games" uses $O(d \log n)$ of
"symbolic set variables" and runs in "quasi-polynomial" time.  
While the dependence on~$n$ is only logarithmic, a natural question is 
whether this dependence is inherent.
Given that $n$ can be prohibitively large in applications, 
reducing dependence on~$n$ is desirable. 
In this section we argue that it is not only possible to eliminate the 
dependence on~$n$ entirely, but it is also possible to exponentially 
improve the dependence on~$d$, resulting in a quasi-polynomial
"symbolic algorithm" for solving "parity games" that uses only $O(\lg d)$
"symbolic set variables".

In the "set-based symbolic model of computation", it is
routine to compute the "attractors" efficiently:
it is sufficient to iterate the controllable predecessor operations.  
Using the results of Jurdzi\'nski and Lazi\'c~\cite{JL17}, one can also 
represent a path of nodes from the root to a leaf in the
tree $\succinct{n, d/2}$ in $O(\lg n \cdot \lg d)$ bits, and for every node
on such a path, to compute its number of children in 
$O(\lg n \cdot \lg d)$ standard primitive operations. 
This allows to run the whole "universal algorithm"
(Algorithm~\ref{algo:universal-parity}) on an $(n, d)$-"small@@game" 
"parity game" and two copies of trees~$\succinct{n, d/2}$, using only 
$O(\lg n \cdot \lg d)$ bits to represent the relevant nodes in the
trees~$\Tc^{\Even}$ and~$\Tc^{\Odd}$ throughout the execution. 

The depth of the tree of recursive calls of the "universal
algorithm" on an $(n, d)$-small parity game is at most~$d$.
Moreover, in every recursive call, only a small constant number of set
variables is needed because only the latest sets $V^{\Gc_i}$,  
$D_i$, $V^{\Gc'_i}$, and~$U_i$ are needed at any time. 
It follows that the overall number of "symbolic set variables" needed  
to run the "universal algorithm" is~$O(d)$. 
Also note that every recursive call can be implemented symbolically
using a constant number of "primitive symbolic operations" and two
symbolic "attractor" computations. 

This improves the "symbolic space" from Chatterjee, Dvo\v{r}\'{a}k,
Henzinger, and Svozil's $O(d \lg n)$
to~$O(d)$, while keeping the running time quasi-polynomial. 
This "symbolic algorithm" is very simple and straightforward
to implement, which makes it particularly promising and attractive
for empirical evaluation and deployment in applications.

\begin{restatable}{theorem}{symbolic}
  \labelwithproof{thm:symbolic}
  There exists a "symbolic algorithm" that solves $(n, d)$-"small@@game" "parity
  games" using $O(\lg d)$ "symbolic set variables", 
  $O(\log d \cdot \log n)$ bits of con\-ven\-tio\-nal space, and whose
  running time is polynomial if $d = O(\log n)$, and
  "quasi-polynomial", namely $n^{2\lg(d/{\lg n})+O(1)}$,
  if $d = \omega(\log n)$.   
\end{restatable}
Using the same arguments, we obtain a "symbolic algorithm" to solve "nested fixpoint equations"
in quasi-polynomial time and $O(\lg d)$ "symbolic space".

\bibliographystyle{alphaurl}
\bibliography{parity.bib}

\newcommand{\etalchar}[1]{$^{#1}$}
\begin{thebibliography}{BKMMP19}

\bibitem[AG11]{AG11}
Krzysztof~R. Apt and Erich Gr{\"{a}}del, editors.
\newblock {\em Lectures in Game Theory for Computer Scientists}.
\newblock Cambridge University Press, 2011.
\newblock URL: \url{http://www.cambridge.org/gb/knowledge/isbn/item5760379}.

\bibitem[ANP21]{ANP21}
Andr{\'{e}} Arnold, Damian Niwi\'nski, and Paweł Parys.
\newblock A quasi-polynomial black-box algorithm for fixed point evaluation.
\newblock In {\em CSL}, volume 183 of {\em LIPIcs}, pages 9:1--9:23. Schloss
  Dagstuhl - Leibniz-Zentrum f{\"{u}}r Informatik, 2021.
\newblock \href {https://doi.org/10.4230/LIPIcs.CSL.2021.9}
  {\path{doi:10.4230/LIPIcs.CSL.2021.9}}.

\bibitem[BCM{\etalchar{+}}92]{BCMDH92}
Jerry~R. Burch, Edmund~M. Clarke, Kenneth~L. McMillan, David~L. Dill, and L.~J.
  Hwang.
\newblock Symbolic model checking: 10{\^{}}20 states and beyond.
\newblock {\em Inf. Comput.}, 98(2):142--170, 1992.
\newblock \href {https://doi.org/10.1016/0890-5401(92)90017-A}
  {\path{doi:10.1016/0890-5401(92)90017-A}}.

\bibitem[BKMMP19]{BKMP21}
Paolo Baldan, Barbara K\"{o}nig, Christina Mika-Michalski, and Tommaso Padoan.
\newblock Fixpoint games on continuous lattices.
\newblock {\em Proc. ACM Program. Lang.}, 3(POPL), January 2019.
\newblock \href {https://doi.org/10.1145/3290339} {\path{doi:10.1145/3290339}}.

\bibitem[BW18]{BW18}
Julian~C. Bradfield and Igor Walukiewicz.
\newblock The mu-calculus and model checking.
\newblock In Edmund~M. Clarke, Thomas~A. Henzinger, Helmut Veith, and Roderick
  Bloem, editors, {\em Handbook of Model Checking}, pages 871--919. Springer,
  2018.
\newblock \href {https://doi.org/10.1007/978-3-319-10575-8_26}
  {\path{doi:10.1007/978-3-319-10575-8_26}}.

\bibitem[CAH11]{CAH11}
Krishnendu Chatterjee, Luca~De Alfaro, and Thomas~A. Henzinger.
\newblock Qualitative concurrent parity games.
\newblock {\em ACM Trans. Comput. Logic}, 12(4), July 2011.
\newblock \href {https://doi.org/10.1145/1970398.1970404}
  {\path{doi:10.1145/1970398.1970404}}.

\bibitem[CDF{\etalchar{+}}19]{CDFJLP19}
Wojciech Czerwi\'nski, Laure Daviaud, Nathanaël Fijalkow, Marcin Jurdzi\'nski,
  Ranko Lazi\'c, and Paweł Parys.
\newblock Universal trees grow inside separating automata: {Q}uasi-polynomial
  lower bounds for parity games.
\newblock In {\em SODA}, pages 2333--2349, 2019.
\newblock \href {https://doi.org/10.1137/1.9781611975482.142}
  {\path{doi:10.1137/1.9781611975482.142}}.

\bibitem[CDHS18]{CDHS18}
Krishnendu Chatterjee, Wolfgang Dvo\v{r}\'{a}k, Monika Henzinger, and Alexander
  Svozil.
\newblock Quasipolynomial set-based symbolic algorithms for parity games.
\newblock In {\em LPAR-22}, volume~57 of {\em EPiC Series in Computing}, pages
  233--253, 2018.
\newblock \href {https://doi.org/10.29007/5z5k} {\path{doi:10.29007/5z5k}}.

\bibitem[CJK{\etalchar{+}}17]{CJKLS17}
Cristian~S. Calude, Sanjay Jain, Bakhadyr Khoussainov, Wei Li, and Frank
  Stephan.
\newblock Deciding parity games in quasipolynomial time.
\newblock In {\em STOC}, pages 252--263, 2017.
\newblock \href {https://doi.org/10.1145/3055399.3055409}
  {\path{doi:10.1145/3055399.3055409}}.

\bibitem[DJL18]{DJL18}
Laure Daviaud, Marcin Jurdzi\'nski, and Ranko Lazi\'c.
\newblock A pseudo-quasi-polynomial algorithm for mean-payoff parity games.
\newblock In {\em Proceedings of the 33rd Annual {ACM/IEEE} Symposium on Logic
  in Computer Science, {LICS} 2018, Oxford, UK, July 09-12, 2018}, pages
  325--334. {ACM}, 2018.
\newblock \href {https://doi.org/10.1145/3209108.3209162}
  {\path{doi:10.1145/3209108.3209162}}.

\bibitem[DJL19]{DJL19}
Laure Daviaud, Marcin Jurdzi\'nski, and Karoliina Lehtinen.
\newblock Alternating weak automata from universal trees.
\newblock In {\em {CONCUR} 2019, August 27-30, 2019}, LIPIcs, pages
  18:1--18:14, 2019.
\newblock \href {https://doi.org/10.4230/LIPIcs.CONCUR.2019.18}
  {\path{doi:10.4230/LIPIcs.CONCUR.2019.18}}.

\bibitem[DJT20]{DJT20}
Laure Daviaud, Marcin Jurdzi\'nski, and K.~S. Thejaswini.
\newblock The strahler number of a parity game.
\newblock In {\em {ICALP} 2020, July 8-11, 2020}, volume 168 of {\em LIPIcs},
  pages 123:1--123:19, 2020.
\newblock \href {https://doi.org/10.4230/LIPIcs.ICALP.2020.123}
  {\path{doi:10.4230/LIPIcs.ICALP.2020.123}}.

\bibitem[EJ91]{EJ91}
E.~Allen Emerson and Charanjit Jutla.
\newblock Tree automata, mu-calculus and determinacy.
\newblock In {\em FOCS}, pages 368--377, 1991.
\newblock \href {https://doi.org/10.1109/SFCS.1991.185392}
  {\path{doi:10.1109/SFCS.1991.185392}}.

\bibitem[EJS93]{EJS93}
E.~Allen Emerson, Charanjit~S. Jutla, and Aravinda~Prasad Sistla.
\newblock On model-checking for fragments of {\(\mu\)}-calculus.
\newblock In Costas Courcoubetis, editor, {\em Computer Aided Verification, 5th
  International Conference, {CAV} '93, Elounda, Greece, June 28 - July 1, 1993,
  Proceedings}, volume 697 of {\em Lecture Notes in Computer Science}, pages
  385--396. Springer, 1993.
\newblock \href {https://doi.org/10.1007/3-540-56922-7_32}
  {\path{doi:10.1007/3-540-56922-7_32}}.

\bibitem[EWS05]{EWS05}
Kousha Etessami, Thomas Wilke, and Rebecca~A. Schuller.
\newblock Fair simulation relations, parity games, and state space reduction
  for {B}\"uchi automata.
\newblock {\em {SIAM} J. Comput.}, 34(5):1159--1175, 2005.
\newblock \href {https://doi.org/10.1137/S0097539703420675}
  {\path{doi:10.1137/S0097539703420675}}.

\bibitem[Fea10]{Fea10}
John Fearnley.
\newblock Exponential lower bounds for policy iteration.
\newblock In {\em ICALP}, volume 6199 of {\em Lecture Notes in Computer
  Science}, pages 551--562. Springer, 2010.
\newblock \href {https://doi.org/10.1007/978-3-642-14162-1_46}
  {\path{doi:10.1007/978-3-642-14162-1_46}}.

\bibitem[FHZ11]{FHZ11}
Oliver Friedmann, Thomas~Dueholm Hansen, and Uri Zwick.
\newblock Subexponential lower bounds for randomized pivoting rules for the
  simplex algorithm.
\newblock In {\em STOC}, pages 283--292. {ACM}, 2011.
\newblock \href {https://doi.org/10.1145/1993636.1993675}
  {\path{doi:10.1145/1993636.1993675}}.

\bibitem[Fri09]{Fri09}
Oliver Friedmann.
\newblock An exponential lower bound for the parity game strategy improvement
  algorithm as we know it.
\newblock In {\em Proceedings of the 24th Annual {IEEE} Symposium on Logic in
  Computer Science, {LICS} 2009, 11-14 August 2009, Los Angeles, CA, {USA}},
  pages 145--156. {IEEE} Computer Society, 2009.
\newblock \href {https://doi.org/10.1109/LICS.2009.27}
  {\path{doi:10.1109/LICS.2009.27}}.

\bibitem[Fri11a]{Fri11r}
Oliver Friedmann.
\newblock Recursive algorithm for parity games requires exponential time.
\newblock {\em {RAIRO} Theor. Informatics Appl.}, 45(4):449--457, 2011.
\newblock \href {https://doi.org/10.1051/ita/2011124}
  {\path{doi:10.1051/ita/2011124}}.

\bibitem[Fri11b]{Fri11}
Oliver Friedmann.
\newblock A subexponential lower bound for {Z}adeh's pivoting rule for solving
  linear programs and games.
\newblock In {\em IPCO}, volume 6655 of {\em Lecture Notes in Computer
  Science}, pages 192--206. Springer, 2011.
\newblock \href {https://doi.org/10.1007/978-3-642-20807-2_16}
  {\path{doi:10.1007/978-3-642-20807-2_16}}.

\bibitem[GTW02]{GTW01}
Erich Gr{\"{a}}del, Wolfgang Thomas, and Thomas Wilke, editors.
\newblock {\em Automata, Logics, and Infinite Games: {A} Guide to Current
  Research [outcome of a Dagstuhl seminar, February 2001]}, volume 2500 of {\em
  Lecture Notes in Computer Science}. Springer, 2002.
\newblock \href {https://doi.org/10.1007/3-540-36387-4}
  {\path{doi:10.1007/3-540-36387-4}}.

\bibitem[HS19]{HS19}
Daniel Hausmann and Lutz Schr{\"o}der.
\newblock Optimal satisfiability checking for arithmetic {$\mu$}-calculi.
\newblock In {\em Foundations of Software Science and Computation Structures},
  pages 277--294. Springer International Publishing, 2019.
\newblock \href {https://doi.org/10.1007/978-3-030-17127-8_16}
  {\path{doi:10.1007/978-3-030-17127-8_16}}.

\bibitem[HS21]{HS21}
Daniel Hausmann and Lutz Schr{\"{o}}der.
\newblock Quasipolynomial computation of nested fixpoints.
\newblock In {\em TACAS}, volume 12651 of {\em Lecture Notes in Computer
  Science}, pages 38--56. Springer, 2021.
\newblock \href {https://doi.org/10.1007/978-3-030-72016-2_3}
  {\path{doi:10.1007/978-3-030-72016-2_3}}.

\bibitem[HSC16]{HSC16}
Ichiro Hasuo, Shunsuke Shimizu, and Corina C{\^{\i}}rstea.
\newblock Lattice-theoretic progress measures and coalgebraic model checking.
\newblock In {\em POPL}, pages 718--732. {ACM}, 2016.
\newblock \href {https://doi.org/10.1145/2837614.2837673}
  {\path{doi:10.1145/2837614.2837673}}.

\bibitem[JL17]{JL17}
Marcin Jurdzi\'nski and Ranko Lazi\'c.
\newblock Succinct progress measures for solving parity games.
\newblock In {\em LICS}, pages 1--9, 2017.
\newblock \href {https://doi.org/10.1109/LICS.2017.8005092}
  {\path{doi:10.1109/LICS.2017.8005092}}.

\bibitem[JPZ08]{JPZ08}
Marcin Jurdzi\'nski, Mike Paterson, and Uri Zwick.
\newblock A deterministic subexponential algorithm for solving parity games.
\newblock {\em SIAM Journal on Computing}, 38(4):1519--1532, 2008.
\newblock \href {https://doi.org/10.1137/070686652}
  {\path{doi:10.1137/070686652}}.

\bibitem[Leh18]{Leh18}
Karoliina Lehtinen.
\newblock A modal $\mu$ perspective on solving parity games in quasi-polynomial
  time.
\newblock In {\em LICS}, pages 639--648, 2018.
\newblock \href {https://doi.org/10.1145/3209108.3209115}
  {\path{doi:10.1145/3209108.3209115}}.

\bibitem[LMS20]{LMS20}
Michael Luttenberger, Philipp~J. Meyer, and Salomon Sickert.
\newblock Practical synthesis of reactive systems from {LTL} specifications via
  parity games.
\newblock {\em Acta Informatica}, 57(1-2):3--36, 2020.
\newblock \href {https://doi.org/10.1007/s00236-019-00349-3}
  {\path{doi:10.1007/s00236-019-00349-3}}.

\bibitem[LPSW22]{LPSW22}
Karoliina Lehtinen, Paweł Parys, Sven Schewe, and Dominik Wojtczak.
\newblock {A Recursive Approach to Solving Parity Games in Quasipolynomial
  Time}.
\newblock {\em {Logical Methods in Computer Science}}, {Volume 18, Issue 1},
  January 2022.
\newblock URL: \url{https://lmcs.episciences.org/8953}, \href
  {https://doi.org/10.46298/lmcs-18(1:8)2022}
  {\path{doi:10.46298/lmcs-18(1:8)2022}}.

\bibitem[LSW19]{LSW19}
Karoliina Lehtinen, Sven Schewe, and Dominik Wojtczak.
\newblock Improving the complexity of {P}arys' recursive algorithm.
\newblock arXiv:1904.11810, April 2019.
\newblock \href {https://doi.org/10.48550/arXiv.1904.11810}
  {\path{doi:10.48550/arXiv.1904.11810}}.

\bibitem[McN93]{McN93}
Robert McNaughton.
\newblock Infinite games played on finite graphs.
\newblock {\em Annals of Pure and Applied Logic}, 65(2):149--184, 1993.
\newblock \href {https://doi.org/10.1016/0168-0072(93)90036-D}
  {\path{doi:10.1016/0168-0072(93)90036-D}}.

\bibitem[Par19]{Par19}
Paweł Parys.
\newblock Parity games: Zielonka's algorithm in quasi-polynomial time.
\newblock In {\em 44th International Symposium on Mathematical Foundations of
  Computer Science, {MFCS} 2019}, 2019.
\newblock \href {https://doi.org/10.4230/LIPIcs.MFCS.2019.10}
  {\path{doi:10.4230/LIPIcs.MFCS.2019.10}}.

\bibitem[Sei96]{S96}
Helmut Seidl.
\newblock Fast and simple nested fixpoints.
\newblock {\em Information Processing Letters}, 59(6):303--308, 1996.
\newblock \href {https://doi.org/10.1016/0020-0190(96)00130-5}
  {\path{doi:10.1016/0020-0190(96)00130-5}}.

\bibitem[vD18]{vDij18}
Tom van Dijk.
\newblock Oink: {A}n implementation and evaluation of modern parity game
  solvers.
\newblock In {\em Tools and Algorithms for the Construction and Analysis of
  Systems, 24th International Conference, TACAS 2018}, volume 10805 of {\em
  LNCS}, pages 291--308, Thessaloniki, Greece, 2018. Springer.
\newblock \href {https://doi.org/10.1007/978-3-319-89960-2_16}
  {\path{doi:10.1007/978-3-319-89960-2_16}}.

\bibitem[VJ00]{VJ00}
Jens V\"oge and Marcin Jurdzi\'nski.
\newblock A discrete strategy improvement algorithm for solving parity games.
\newblock In {\em CAV 2000}, volume 1855 of {\em LNCS}, pages 202--215,
  Chicago, IL, USA, 2000. Springer.
\newblock \href {https://doi.org/10.1007/10722167_18}
  {\path{doi:10.1007/10722167_18}}.

\bibitem[Zie98]{Zie98}
Wieslaw Zielonka.
\newblock Infinite games on finitely coloured graphs with applications to
  automata on infinite trees.
\newblock {\em Theoretical Computer Science}, 200:135--183, 1998.
\newblock \href {https://doi.org/10.1016/S0304-3975(98)00009-7}
  {\path{doi:10.1016/S0304-3975(98)00009-7}}.

\end{thebibliography}

\clearpage\appendix
\section{McNaughton-Zielonka algorithm}
\label{sec:McNZ}

\begin{algorithm}[htb]
    \AP\DontPrintSemicolon
    \SetKwFunction{ZsolveEven}{$\text{McN-Z}_{\Even}$}
    \SetKwFunction{ZsolveOdd}{$\text{McN-Z}_{\Odd}$}
    \SetKwProg{fun}{procedure}{:}{}
    \fun{\ZsolveEven{$\Gc, d$}}{
    \If{$d=0$}{\Return{$V^{\Gc}$}}
      $i \leftarrow 0; \; \Gc_1 \leftarrow \Gc$\;
      \Repeat{$U_i = \emptyset$}{
        $i \leftarrow i+1$\;
        $D_i \leftarrow \pi^{-1}(d) \cap \Gc_i$\;
        $\Gc_i' \leftarrow \Gc_i \setminus \Attr_{\Even}^{\Gc_i}(D_i)$\;
        $U_i \leftarrow \ZsolveOdd\bigl(\Gc_i', d-1
        \bigr)$\;
        $\Gc_{i+1} \leftarrow \Gc_i \setminus \Attr^{\Gc_i}_{\Odd}(U_i)$\;
      }
    \Return{$V^{\Gc_i}$}
    }
    \fun{\ZsolveOdd{$\Gc, d$}}{
      $i \leftarrow 0; \; \Gc_1 \leftarrow \Gc$\;
      \Repeat{$U_i = \emptyset$}{
        $i \leftarrow i+1$\;
        $D_i \leftarrow \pi^{-1}(d) \cap \Gc_i$\;
          $\Gc_i' \leftarrow \Gc_i \setminus \Attr_{\Odd}^{\Gc_i}(D_i)$\;
        $U_i \leftarrow \ZsolveEven\bigl(\Gc_i', d-1
          \bigr)$\;
        $\Gc_{i+1} \leftarrow \Gc_i \setminus \Attr^{\Gc_i}_{\Even}(U_i)$\;
      }
      \Return{$V^{\Gc_i}$}
    }
    \caption{
      \label{algo:Z-Solve}
      \emph{""McNaughton-Zielonka algorithm""}}
  \end{algorithm}
  
  The classic recursive McNaughton-Zielonka algorithm
  (Algorithm~\ref{algo:Z-Solve}) computes the 
  largest dominia in a parity game.
  In order to obtain the largest Even dominion in a parity game~$\Gc$,
  it suffices to call $\ZsolveEven(\Gc, d)$, where $d$ is even and all
  vertex priorities in~$\Gc$ are at most~$d$.
  In order to obtain the largest Odd dominion in a parity game~$\Gc$,
  it suffices to call $\ZsolveOdd(\Gc, d)$, where $d$ is odd and all
  vertex priorities in~$\Gc$ are at most~$d$.
  
  The procedures $\ZsolveEven$ and $\ZsolveOdd$ are mutually recursive
  and whenever a recursive call is made, the second argument~$d$
  decreases by~$1$. 
  Figure~\ref{figure:McN-Z-iter} illustrates one iteration of the main
  loop in a call of procedure $\ZsolveEven$.
  The outer rectangle denotes subgame~$\Gc_i$, the thin horizontal
  rectangle at the top denotes the set $D_i$ of the vertices in~$\Gc_i$
  whose priority is~$d$, and the set below the horizontal wavy line is 
  subgame~$\Gc_i'$, which is the set of vertices in~$\Gc_i$ that are not
  in the attractor~$\Attr_{\Even}^{\Gc_i}(D_i)$. 
  The recursive call of $\ZsolveOdd$ returns the set~$U_i$, and
  $\Gc_{i+1}$ is the subgame to the left of the vertical zig-zag line,
  and it is induced by the set of vertices in~$\Gc_i$ that are not in
  the attractor~$\Attr_{\Odd}^{\Gc_i}(U_i)$. 
  
  A way to prove the correctness of McNaughton-Zie\-lon\-ka algorithm we 
  wish to highlight here is to enhance the algorithm slightly to produce
  not just a set of vertices but also an Even attractor decomposition of
  the set and an Odd attractor decomposition of its complement.   
  We explain how to modify procedure $\ZsolveEven$ and leave it as an
  exercise for the reader to analogously modify procedure $\ZsolveOdd$. 
  In procedure $\ZsolveEven(\Gc, d)$, replace the line  
  \[
    U_i \leftarrow \ZsolveOdd(\Gc_i', d-1)
  \]
  by the line
  \[
    U_i, \Hc_i, \Hc_i' \leftarrow \ZsolveOdd(\Gc_i', d-1)\,.
  \]
  Moreover, if upon termination of the {\bf repeat-until} loop we have 
  \[
    \Hc_i \: = \:
    \seq{\emptyset, (S_1, \Ic_1, A_1), \dots, (S_k, \Ic_k, A_k)}
  \]
  then instead of returning just the set~$V^{\Gc_i}$, let the procedure 
  return both $V^{\Gc_i}$ and the following two objects:
  \begin{equation}
  \label{eq:even}
    \seq{\Attr_{\Even}^{\Gc_i}(D_i), (S_1, \Ic_1, A_1), \dots, (S_k, \Ic_k, A_k)}
  \end{equation}
  and 
  \begin{equation}
  \label{eq:odd}
    \seq{\emptyset, \bigl(U_1, \Hc_1', \Attr_{\Odd}^{\Gc_1}(U_1)\bigr), \dots, 
      \bigl(U_i, \Hc_i', \Attr_{\Odd}^{\Gc_i}(U_i)\bigr)}
  \end{equation}
  

  In an inductive argument by induction on~$d$ and~$i$, the inductive
  hypothesis is that:
  \begin{itemize}
  \item
    $\Hc_i'$ is an Odd $(d-1)$-attractor decomposition of the subgame
    $\Gc'_i \cap U_i$;
  \item
    $\Hc_i$ is an Even $d$-attractor decomposition of the
    subgame~$\Gc'_i \setminus U_i$;
  \end{itemize}
  and the inductive step is then to show that:
  \begin{itemize}
  \item  
    for every~$i$, (\ref{eq:odd}) is an Odd $(d+1)$-attractor
    decomposition of subgame~$\Gc \setminus \Gc_{i+1}$;
  \item
    upon termination of the \textbf{repeat-until} loop, (\ref{eq:even})
    is an Even $d$-attractor decomposition of subgame~$\Gc_{i+1}$.  
  \end{itemize}
  The general arguments in such a proof are well
  known~\cite{McN93,Zie98,JPZ08,DJL18}  
  and hence we omit the details here.  
  
\section{Embeddable decomposition Theorem}

\begin{proofappendix}{thm:embeddable-decomposition}{\embeddabledecomposition}
  Without loss of generality, assume that $d$ is even and 
  \[
  \Hc = \seq{A, (S_1, \Hc_1, A_1), \dots, (S_k, \Hc_k, A_k)}
  \]
  is an Even $d$-attractor decomposition of~$\Gc$,
  where $A$ is the Even attractor to the set~$D$ of vertices of
  priority~$d$ in~$\Gc$.
  In Figure~\ref{figure:embeddable-decomposition-proof}, set~$T$ and
  the subgame $\Gc'$ it induces form the pentagon obtained from the
  largest rectangle by removing the triangle above the diagonal line in
  the top-left corner. 
  Sets $A$, $S_1$, and $A_1$ are also illustrated, together with 
  sets $A'$, $S_1'$, $A_1'$ and subgames $\Gc_1$, $\Gc_2$,
  $\Gc_1'$, and $\Gc_2'$, which are defined as follows.

  Let $\Gc_1 = \Gc \setminus A$, and $\Gc_2 = \Gc_1 \setminus A_1$.
  We will define sets $A'$, $S_1'$, $A_1'$, \dots, $S_\ell'$, $A_\ell'$,
  and Even $(d-2)$-attractor decompositions $\Hc_1', \dots, \Hc_\ell'$ 
  of subgames $\Gc \cap S_1'$, \dots, $\Gc \cap S_\ell'$,
  respectively, such that
  \[
  \Hc'= \seq{A', (S_1', \Hc_1', A_1'), \dots, (S_k', \Hc_\ell', A_\ell')}
  \]
  is an Even $d$-attractor decomposition of subgame~$\Gc'$ and  
  $\Tc_{\Hc}$ embeds~$\Tc_{\Hc'}$.

  Let $A'$ be the Even attractor to $D \cap T$ in~$\Gc'$
  and let $\Gc_1' = \Gc' \setminus A'$.
  Set $S_1' = S_1 \cap \Gc_1'$,
  let $A_1'$ be the Even attractor to~$S_1'$ in~$\Gc_1'$, and
  let $\Gc_2' = \Gc_1' \setminus A_1'$.

  Firstly, since $D \subseteq V^{\Gc}$ and $T$ is a trap for Even
  in~$\Gc$, by Proposition~\ref{prop:attractor-inside-trap}, we have
  that $\Gc_1'$ is a trap for Even in subgame~$\Gc_1$.
  Since $S_1 \subseteq V^{\Gc_1}$ and subgame $\Gc_1'$ is a trap for
  Even in subgame~$\Gc_1$, again by
  Proposition~\ref{prop:attractor-inside-trap}, we conclude that
  $\Gc_2'$ is a trap for Even in subgame~$\Gc_2$.

  Secondly, we argue that $S_1'$ is an Even dominion in
  subgame~$\Gc_1'$.
  This follows by recalling that $S_1$ is a dominion for Even
  in~$\Gc_1$ and $\Gc_1'$ is a trap for Even in~$\Gc_1$, and then
  applying Proposition~\ref{prop:dominion-cap-trap}.

  Thirdly, we argue that $S_1'$ is a trap for Even in
  subgame~$\Gc \cap S_1$.
  This follows by recalling that $S_1$ is a trap for Odd in~$\Gc_1$
  and that $\Gc_1'$ is a trap for Even in $\Gc_1$, and then applying
  Proposition~\ref{prop:dominion-cap-trap}.
  
  \begin{figure}
    \centering
    \makebox[\textwidth][c]{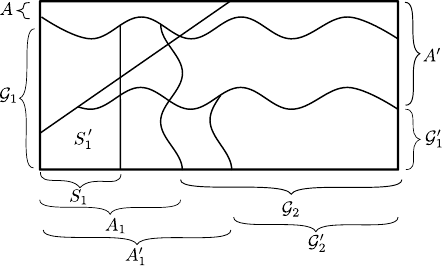}
    \caption{Attractors, subgames, and dominia in the proof of the
      embeddable decomposition theorem.}  
    \label{figure:embeddable-decomposition-proof}
  \end{figure}

  We are now in a position to apply the inductive hypothesis twice in
  order to complete the definition of the attractor
  decomposition~$\Hc'$.
  Firstly, recall that $S_1'$ is a trap for Even in subgame
  $\Gc \cap S_1$ and that $\Hc_1$ is a $(d-2)$-attractor
  decomposition of~$\Gc \cap S_1$, so we can apply the inductive
  hypothesis to obtain a $(d-2)$-attractor decomposition~$\Hc_1'$
  of subgame~$\Gc \cap S_1'$,
  such that $\Tc_{\Hc_1}$ embeds~$\Tc_{\Hc_1'}$.
  Secondly, note that
  \[
    \Ic \: = \: 
    \seq{\emptyset, (S_2, \Hc_2, A_2), \dots, (S_k, \Hc_k, A_k)}
  \]
  is a $d$-attractor decomposition of~$\Gc_2$.
  We find a $d$-attractor decomposition~$\Ic'$
  of subgame~$\Gc_2'$,
  such that~$\Tc_{\Ic}$ embeds~$\Tc_{\Ic'}$.
  Recalling that $\Gc_2'$ is a trap for Even in subgame~$\Gc_2$, it
  suffices to use the inductive hypothesis for subgame $\Gc_2'$ of
  game $\Gc_2$ and the $d$-attractor decomposition~$\Ic$
  of~$\Gc_2$.

  Verifying that $\Hc'$ is a $d$-attractor decomposition of~$\Gc'$
  is routine.
  That $\Tc_{\Hc}$ embeds~$\Tc_{\Hc'}$
  also follows routinely from
  $\Tc_{\Hc_1}$ embedding~$\Tc_{\Hc_1'}$ and
  $\Tc_{\Ic}$ embedding~$\Tc_{\Ic'}$.
\end{proofappendix}
\section{Dominion separation theorem}

Before we prove the dominion separation theorem: we recall 
a simple proposition from Lehtinen, Parys, Schewe and Wojtczak~\cite{LPSW22}. 
Note that it is a straightfoward corollary of the dual of
Proposition~\ref{prop:attractor-inside-trap}
(in case $B \cap T = \emptyset$). 

\begin{proposition}
  \label{prop:trap-safe-from-attractor}
  If $T$ is a trap for Odd in~$\Gc$ and $T \cap B = \emptyset$ then
  we also have that $T \cap \Attr_{\Odd}^{\Gc}(B) = \emptyset$.
\end{proposition}

\begin{proofappendix}{thm:dominion-separation}{\dominionseparation}
  We prove the statement
  of part~\ref{item:dominion-separation-a};
  the proof of part~\ref{item:dominion-separation-b} is analogous. 

  The proof is by induction on the height of 
  tree~$\Tc^{\Odd} \interleaving \Tc^{\Even}$
  (the ``outer'' induction).
  If the height is~$0$ then tree $\Tc^{\Odd}$ is the trivial
  tree~$\seq{}$; 
  hence $k = 0$, the algorithm returns the set $V^{\Gc_1} = V^{\Gc}$,
  which contains the largest Even dominion, and which is trivially
  disjoint from the largest Odd dominion 
  (because the latter is empty). 

  If the height of $\Tc^{\Odd} \interleaving \Tc^{\Even}$ is positive, then
  we split the proof of the separation property into two parts.

\begin{figure}
  \centering
  \makebox[\textwidth][c]{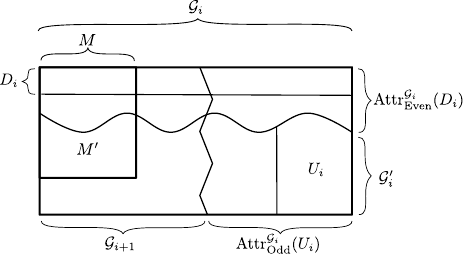}
  \caption{Attractors, subgames, and dominia in the first part of the
    proof of the dominion separation theorem.} 
  \label{figure:dominion-separation-even}
\end{figure}

  \paragraph*{Even dominia embedded by $\Tc^{\Even}$ are included
    in~$\Gc_{i+1}$.}     
  We prove by induction on~$i$ (the ``inner'' induction) that for 
  $i = 0, 1, 2, \dots, k$, if $M$ is an Even dominion in~$\Gc$ that 
  $\Tc^{\Even}$ embeds, then $M \subseteq \Gc_{i+1}$. 

  For $i = 0$, this is moot because $\Gc_1 = \Gc$. 
  
  For $i > 0$, let $M$ be an Even dominion that has an Even
  $d$-attractor decomposition~$\Hc$ such that $\Tc^{\Even}$
  embeds~$\Tc_{\Hc}$. 
  The inner inductive hypothesis (for~$i-1$) implies that 
  $M \subseteq \Gc_i$.  

  The reader is encouraged to systematically refer to
  Figure~\ref{figure:dominion-separation-even} to better follow the
  rest of this part of the proof.  
  
  Let $M' = M \setminus \Attr_{\Even}^{\Gc_i}(D_i)$. 
  Because $\Gc_i \setminus \Attr_{\Even}^{\Gc_i}(D_i)$ is a
  trap for Even in $\Gc_i$ and $M$ is a trap for Odd in~$\Gc_i$,
  the dual of Proposition~\ref{prop:dominion-cap-trap} yields that $M'$
  is a trap for Even in $\Gc_i \cap M$.

  Then, because $\Hc$ is an Even $d$-attractor decomposition of
  $\Gc \cap M$, it follows by
  Theorem~\ref{thm:embeddable-decomposition} that there is an Even  
  $d$-attractor decomposition $\Hc'$ of $\Gc_i \cap M'$
  such that $\Tc_{\Hc}$ embeds~$\Tc_{\Hc'}$, and hence also
  $\Tc^{\Even}$ embeds~$\Tc_{\Hc'}$. 

  Therefore, because $M'$ is an Even dominion in the game
  $\Gc_i \setminus \Attr_{\Even}^{\Gc_i}(D_i)$,
  part~\ref{item:dominion-separation-b} of the outer inductive
  hypothesis 
  yields $M' \cap U_i = \emptyset$.

  Finaly, because 
  $M \setminus M' \subseteq \Attr_{\Even}^{\Gc_i}(D_i)$
  and $(M' \setminus M) \cap U_i = \emptyset$, 
  it follows that $M \cap U_i = \emptyset$. 
  By Proposition~\ref{prop:trap-safe-from-attractor},
  we obtain $M \cap \Attr_{\Odd}^{\Gc_i}(U_i) = \emptyset$
  and hence $M \subseteq \Gc_{i+1}$. 
  
\begin{figure}
  \centering
  \makebox[\textwidth][c]{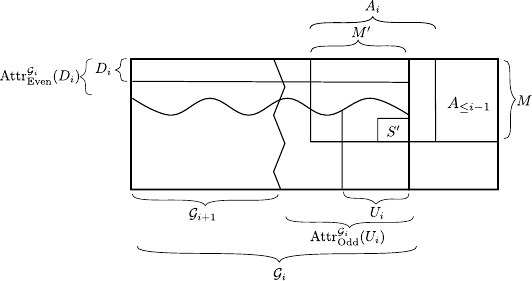}
  \caption{Attractors, subgames, and dominia in the second part of the
    proof of the dominion separation theorem.} 
  \label{figure:dominion-separation-odd}
\end{figure}

  \paragraph*{Odd dominia embedded by 
    $\seq{\Tc_1^{\Odd}, \hdots, \Tc_i^{\Odd}}$ are disjoint
    from~$\Gc_{i+1}$.}   
  We prove by induction on~$i$ (another ``inner'' induction)
  that for $i = 0, 1, \dots, k$, if $M$ is an Odd dominion in~$\Gc$
  that $\seq{\Tc^\Odd_1,\hdots,\Tc^\Odd_i}$ embeds, then 
  $\Gc_{i+1} \cap M = \emptyset$. 

  For $i = 0$, note that 
  $\seq{\Tc_1^{\Odd}, \dots, \Tc_i^{\Odd}} = \seq{}$ 
  and the only Odd dominion~$M$ in~$\Gc$ that has an Odd 
  $(d+1)$-attractor decomposition whose tree is the trivial
  tree~$\seq{}$ is the empty set, and hence $\Gc_1 \cap M =
  \emptyset$, because $\Gc_1 = \Gc$.  

  The reader is encouraged to systematically refer to
  Figure~\ref{figure:dominion-separation-odd} to better follow the
  rest of this part of the proof.  

\newcommand{\ddi}{\bar{\imath}}

  For $i > 0$, let 
  \[
    \Hc = \seq{\emptyset, 
      (S_1, \Hc_1, A_1), \dots, (S_{\ddi}, \Hc_{\ddi}, A_{\ddi})}  
  \]
  be an Odd $(d+1)$-attractor decomposition of $\Gc \cap M$ such that  
  $\seq{\Tc^\Odd_1, \dots, \Tc^\Odd_i}$ embeds $\Tc_{\Hc}$. 
  Note that the embedding implies that~$\ddi \leq i$. 

  If $\seq{\Tc_1^{\Odd}, \dots, \Tc_{i-1}^{\Odd}}$ embeds $\Tc_\Hc$ 
  then the inner inductive hypothesis (for~$i-1$) implies that
  $\Gc_i \cap M = \emptyset$ and thus
  $\Gc_{i+1} \cap M = \emptyset$ since $\Gc_{i+1} \subseteq \Gc_i$.

  Otherwise, it must be the case that
  \begin{equation}
  \label{eq:last-subtrees-embed}
    \text{$\Tc^{\Odd}_i$ embeds $\Tc_{\Hc_{\ddi}}$}\,.
  \end{equation}
  Observe that the set 
  $A_{\leq {\ddi}-1} = A_1 \cup A_2 \cup \cdots \cup A_{{\ddi}-1}$
  is a trap for Even in~$\Gc \cap M$, and hence by trap transitivity
  it is a trap for Even in~$\Gc$ because~$M$ is a trap for Even
  in~$\Gc$.  
  Moreover, subgame $\Gc \cap A_{\leq {\ddi}-1}$  
  has an Odd $(d+1)$-attractor decomposition 
  \[
    \Ic \: = \:
    \seq{\emptyset, (S_1, \Hc_1, A_1), \dots, 
      (S_{\ddi-1}, \Hc_{\ddi-1}, A_{\ddi-1})}
  \]
  in~$\Gc$ and hence---by the dual of
  Proposition~\ref{prop:decomposition-dominion-even}---it is an Odd
  dominion in~$\Gc$, and ordered tree  
  $\seq{\Tc_1^{\Odd}, \dots, \Tc_{i-1}^{\Odd}}$ embeds $\Tc_\Ic$. 
  Hence, the inner inductive hypothesis (for~$i-1$) yields
  \begin{equation}
  \label{eq:A-leq-i-prime-minus-1}
    \Gc_i \cap A_{\leq \ddi-1} \: = \: \emptyset \,.
  \end{equation}

  Set $M' = \Gc_i \cap M$ and note that not only 
  $M' \subseteq A_{\ddi}$, but also $M'$ is a trap for Odd
  in~$A_{\ddi}$, because $\Gc_i$ is a trap for Odd in~$\Gc$. 
  Moreover---by Proposition~\ref{prop:dominion-cap-trap}---$M'$
  is an Odd dominion in~$\Gc_i$ because $\Gc_i$ is a trap for Odd
  in~$\Gc$ and $M$ is a dominion for Odd in~$\Gc$.  

  Observe that
    $\Jc =
    \seq{\emptyset, (S_{\ddi}, \Hc_{\ddi}, A_{\ddi})}$
  is an Odd $(d+1)$-attractor decomposition of $\Gc \cap A_{\ddi}$. 
  By the embeddable decomposition theorem 
  (Theorem~\ref{thm:embeddable-decomposition}), it follows that there 
  is an Odd $(d+1)$-attractor decomposition $\Kc$ of $\Gc \cap M'$
  such that $\Tc_{\Jc}$ embeds~$\Tc_{\Kc}$.
  Because of this embedding, $\Kc$ must have the form 
    $\Kc = \seq{\emptyset, (S', \Kc', M')}$.
  Since $\Tc_\Jc$ embeds $\Tc_{\Kc}$, we also have that 
  $\Tc_{\Hc_{\ddi}}$ embeds $\Tc_{\Kc'}$,
  and hence---by (\ref{eq:last-subtrees-embed})---$\Tc^\Odd_i$ 
  embeds~$\Tc_{\Kc'}$. 

  Note that $S'$ is a trap for Odd in $\Gc \cap M'$
  in which every vertex priority is at most~$d-1$, because $\Kc$ is 
  an Odd $(d+1)$-attractor decomposition of~$\Gc \cap M'$. 
  It follows that $S'$ is also an Odd dominion in 
  $\Gc_i \setminus \Attr_{\Even}^{\Gc_i}(D_i)$. 

  The outer inductive hypothesis 
  then yields $S' \subseteq U_i$.
  It follows that 
  \[
    M' = \Attr_{\Odd}^{\Gc_i \cap M'}(S')
    \subseteq \Attr^{\Gc_i}_\Odd(S') 
    \subseteq \Attr^{\Gc_i}_\Odd(U_i) \,,
  \]
  where the first inclusion holds because $M'$ is a trap for Even
  in~$\Gc_i$, and the second follows from monotonicity of the
  attractor operator. 
  When combined with with~(\ref{eq:A-leq-i-prime-minus-1}), this
  implies $\Gc_{i+1} \cap M = \emptyset$. 
\end{proofappendix}
\section{Floweriness lemma}

In this Appendix, we will prove Lemma~\ref{lemma:flowery} and Theorem~\ref{thm:nfe-complexity}.
Whenever we want to denote the fixpoint obtained by repeated application of a monotone function $f$ on a set, we call this $f^*$.
Before we embark on the proofs, we would like to call attention to following property of flowery subgames. It shows how complements of two specific kinds of flowery sets result in another flowery subgame. We will use this property in several of our proofs.
\begin{property}\label{prop:compl}
For $A \subseteq Y\subseteq X \subseteq (U\times [d])$, we have:
$$\flower(X,Y) \setminus  \flower(X, A) = \flower(X\setminus A, Y\setminus A).$$
Notice that $\flower(X\setminus A, Y\setminus A) = \flower(Z\cup W, W)$, where $Z = X\setminus Y$ and $W = X\setminus A$.
\end{property}

Consider the following proposition useful in the proof of the Lemma~\ref{lemma:flowery}.

\begin{proposition}
Given a fixpoint game $\Gc_f$, after removing the Even attractor to the set of Odd vertices with no outgoing edges and the Odd attractor to the Even vertices with no outgoing edges, we are left with a flowery subgame. 
\end{proposition}
\begin{proof}
The game $\Gc_f$ contains exactly the vertices in
the subgame $\flower(U\times [d], U\times [d])$ along with $v_\emptyset$. 
\begin{itemize}
    \item Initially, we remove the only Odd vertex with no outgoing edge: $v_\emptyset$, along with its Even attractor. 
The Even attractor to $v_\emptyset$ in $\Gc_f$ is exactly all the vertices of the flowery set 
$\flower(Z,Z)$ and $v_\emptyset$, where $Z = f^*(\emptyset)$ winning for Even.
The remaining subgame after removing these vertices is the flowery subgame $\flower(U\times [d],\bigl(U\times [d]\bigr)\setminus Z)$ from Property~\ref{prop:compl}.
\item Let us call the flowery subgame obtained from the above procedure $\flower(X,Y)$. Observe that if $Y\subseteq f(X)$, then there is always an outgoing edge for each vertex in the subgame. If not, we remove the Odd attractor to the set of Even vertices with no outgoing edges: $Y\cap \overline{f(X)}$. The complement of this Odd attractor turns out to be the flowery subgame $\flower(X, Y\setminus \overline{f(X)})$ from Property~\ref{prop:compl}.\qedhere
\end{itemize}
\end{proof}
Assuming now that we always have outgoing edges in flowery subgames, we consider the following Lemma which shows how we can compute attractors to sets in these subgames with at most $d\cdot |U|$ many calls to the function $f$.

Let us now prove Lemma~\ref{lemma:flowery} by instead proving a stronger statement stated in Lemma~\ref{lemma:subgames}. To lead to the proof of Lemma~\ref{lemma:subgames}, we need Lemma~\ref{lemma:attractorNFE} which states intuitively that computing attractors to specific flowery subgames lead to specific flowery subgames whose complement is also flowery. 
\begin{lemma}
In a flowery subgame $\Gc = \flower(X,Y)$:
\begin{itemize}
     \item[(a)]the Even attractor to a set of Even vertices $A\subseteq Y$  in $\Gc = \flower(X,Y)$ where $Z = X\setminus Y$ is $$\flower(Z\cup Pre^*_{\Gc, \Even}(A),Pre^*_{\Gc, \Even}(A))$$ 
    where $Pre_{\Gc, \Even}(A) = \bigl(f(Z\cup A)\cap Y\bigr)\cup A$;
    \item[(b)] the Odd attractor to a set of Even vertices $A$ or a subgame $\flower(X,A)$ in $\Gc = \flower(X,Y)$ is $$\flower(X,Pre^*_{\Gc, \Odd}(A))$$ 
    where $Pre_{\Gc, \Odd}(A) = (\overline{ f(X\setminus A)}\cap Y)\cup A$.
\end{itemize}
\label{lemma:attractorNFE}
\end{lemma}

We will break down our Lemma into Propositions~\ref{prop:attractorNFEEven} and~\ref{prop:attractorNFEOdd} which will result in  Corollary~\ref{cor:subgameNFEattractor} from which Lemma~\ref{lemma:attractorNFE} follows.

\begin{proposition}
In a flowery subgame $\Gc = \flower(X,Y)$ and $A\subseteq Y$,
the flowery subgame $\flower(Z\cup Pre_{\Gc,\Even}(A), Pre_{\Gc,\Even}(A))$ is exactly the set of vertices from which Even has a strategy to visit $A$ in at most three steps, where $Pre_{\Gc, \Even}(A) = \bigl(f(Z\cup A)\cap Y\bigr)\cup A$.
\label{prop:attractorNFEEven}
\end{proposition}

\begin{proof}
 We will argue about vertices from which Even has a strategy to visit vertices in $A$ in at most one, two and three steps below.
    \begin{itemize}
        \item[(1)] Consider any Odd vertex $v_B$ where  $B \subseteq Z\cup A$ and the intersection of $B$ with $A$ is non empty. From such a $v_B$, in one step, Even can ensure that a play reaches $A$. All such vertices $v_B$ along with the core $A$ is exactly denoted by the vertices of the subgame $\flower(Z\cup A,A)$.
        \item[(2)] We will show that from any Even vertex $(u,i)\in Pre_{\Gc, \Even}(A) =  f(Z\cup A)\cap Y\}\cup A$, there is a strategy for Even to reach a Even vertex in $A$ in at most two steps. To show this, we will show that:
    \begin{itemize}
        \item[$(\Rightarrow)$] in one step, Even can move to some Odd vertex $v_B\in \flower(Z\cup A,A)$;
        \item[$(\Leftarrow)$] from vertices not in $Pre_{\Gc, \Even}(A)$, all of Even's outgoing edges lead to a vertex not in $\flower(Z\cup A,A)$.
    \end{itemize}
    
    To show the forward direction, let $(u,i)\in (f(Z\cup A)\cap Y)\cup A$, if $(u,i)\in A$ then we are done, if not, the strategy for Even from $(u,i)$ is to choose the Odd vertex $v_{Z\cup A}$ and such an edge exists since $(u,i)\in f(Z\cup A)$, and this Odd vertex is in the flowery subgame $\flower(Z\cup A,A)$.
    
    To show the reverse direction,
    Consider $(u,i)\notin f(Z\cup A)\cup A$ but $(u,i)\in Y$. 
    All edges out of the Even vertex $(u,i)$ leads to an  Odd vertex $v_B$ in $\flower(X,Y)$ such that $B$ has some element other than from $Z$ or $A$ i.e,
    $B\setminus (Z\cup A) \neq \emptyset$. This follows from the monotonicity of $f$ along with our assumption that $(u,i)\notin f(Z\cup A)$. 
    After one step, the game is at an Odd vertex $v_B$ 
    that it is not in $\flower(Z\cup A,A)$.
        \item[(3)]The argument to conclude that $\flower(X,Pre_{\Gc, \Odd}(A))$ is exactly the set we desire is similar to (1).\qedhere
    \end{itemize}
\end{proof}
\begin{proposition}
In a flowery subgame $\Gc = \flower(X,Y)$ and $A\subseteq Y$,
 From any vertex of the flowery subgame $\flower(X,Pre_{\Gc, \Odd}(A))$, Odd has a strategy to visit a set of Even vertices $A$ in at most three steps 
    where $Pre_{\Gc, \Odd}(A) = (\overline{ f(X\setminus A)}\cap Y)\cup A$. The above subgame is the exact set of vertices from which Odd has such a strategy.
\label{prop:attractorNFEOdd}
\end{proposition}
\begin{proof}
 We show the set of vertices from which Even has a strategy to visit vertices in $A$ in at most one, two and three steps below.
    \begin{itemize}
        \item[(1)]From vertex $v_B$ where $B$ of $X$ which intersects with $A$ non-trivially, Odd would be able to reach a vertex in $A$ in at most one step. This exactly is all the Odd vertices in the flowery subgame $\flower(X,A)$. 
        \item[(2)]We will show that in one step, Odd has a strategy to visit the subgame $\flower(X,A)$ from vertices in $Pre_{\Gc, \Odd}(A)\cup A$. We do this by showing inclusion in two direction. 
        
        \begin{itemize}
            \item[$(\Rightarrow)$]Consider $(v,j)\in  Pre_{\Gc, \Odd}(A)  = (\overline{ f(X\setminus A)}\cap Y)\cup A$. If $(v,j)\notin A$, then $(v,j)\in Y$ and $\overline{f(X\setminus A)}$. Mainly note that $(v,j)\notin f(X\setminus A)$. Since all subgames are such that there is always an outgoing edge and given that $f$ is monotone, any Odd vertex $v_B$ in $\flower(X,Y)$ which has an edge to it from $(v,j)$ must be such that 
    $B\cap A\neq \emptyset$. For any choice successors from $(v,j)$ of Even will lead to a vertex $B$ which intersects with $A$ and hence there is a strategy for Odd to move to a vertex in $(u,i)$ in $B\cap A$. 
            \item[$(\Leftarrow)$]  Now we need to show a strategy for Even to remain in the complement of the game $\flower(X,Pre_{\Gc, \Odd}(A))$ for two steps from all other Even vertices. Let us denote $Pre_{\Gc, \Odd}(A)$ by $W$. Note that the complement of  $\flower(X,W)$ in $\flower(X,Y)$ is $\flower(X\setminus W, Y\setminus W)$. Notice that 
    $$Y\setminus W = Y\setminus\bigl( \overline{ f(X\setminus A)}\cup A\bigr)$$ 
    So, any $(w,j)\in Y\setminus Z$ is in $Y$ and since $(w,j)\notin W$, $(w,j)\in f(X\setminus A)$. This means that from any such $(w,j)$, Even can choose the vertex $v_B$ in  $\flower(X\setminus W, Y\setminus W)$ where $B \subseteq X\setminus A$, making sure that in the next step Odd will not be able to take the play to an Even vertex in $A$. 
        \end{itemize}
        \item[(3)]From the structure of the game, it is easy to see that any $v_B$ such that $B$ intersects with $Pre_{\Gc, \Odd}(A)\cup A$ would be able to visit an element in $Pre_{\Gc, \Odd}(A)\cup A$, which we have shown is exactly the set of vertices from which Odd could force the play in at most two steps to visit $A$.\qedhere 
    \end{itemize}
\end{proof}
From the proof of the Propositions~\ref{prop:attractorNFEEven} and~\ref{prop:attractorNFEOdd}, we can extend these to show the following Corollary from which Lemma~\ref{lemma:attractorNFE} follows.
\begin{corollary}\label{cor:subgameNFEattractor}
In a flowery subgame $\Gc = \flower(X,Y)$ and $A\subseteq Y$,
\begin{itemize}
    \item The flowery subgame $\flower(Z\cup Pre_{\Gc,\Even}(A), Pre_{\Gc,\Even}(A))$ is the set of vertices from which Even has a strategy to visit the vertices in $\flower(Z\cup A, A)$ in at most two steps, where $Pre_{\Gc, \Even}(A) = \bigl(f(Z\cup A)\cap Y\bigr)\cup A$;
    \item  The vertices of $\flower(X,Pre_{\Gc, \Odd}(A))$ is the set of vertices from which Odd has a strategy to visit a vertex in $\flower(X,A)$ in at most two steps.
\end{itemize}
\end{corollary}
We state that Lemma~\ref{lemma:attractorNFE} follows naturally from Corollary~\ref{cor:subgameNFEattractor} and conclude the proof of Lemma~\ref{lemma:attractorNFE}.
We will now proceed to the main proof of the section:
\recall{\flowerylemma}

We will instead prove a stronger version of Lemma~\ref{lemma:flowery}, stated below:
\begin{lemma}\label{lemma:subgames}
\begin{itemize}
    \item[$(i)$] If \solveEven is run on a flowery subgame $\flower(X,Y)$,
    then in all iterations in the for-loop in the subgame $\Gc_i$ is of the form $\flower(X\setminus A_i',Y\setminus A_i')$ for $A_i'\subseteq Y$, in particular $\Gc_{k+1}$, which is the set of vertices returned. 
    \item[$(ii)$] If \solveOdd is run on a flowery subgame $\flower(X,Y)$,
    then in all iterations in the for-loop, the subgame $\Gc_i$ is  of the form $\flower(X,Y\setminus A_i')$, where $A_i'\subseteq Y$ in particular $\Gc_{k+1}$, which is the set of vertices returned. 
\end{itemize}
\end{lemma}

\begin{proof}[Proof of \Cref{lemma:subgames,lemma:flowery}]
    \label{proof-lemma:flowery}
We will prove this by induction on the sum of the number of vertices  in these subgames and the number of vertices on which these calls are made.
For the base case, with an empty set irrespective of any priority, the above statement is trivially true.
We will now prove that $(i)$ and $(ii)$ hold for games with at least one vertex and trees $\Tc_{\Even}$ and $\Tc_{\Odd}$. The proof follows from Lemma~\ref{lemma:attractorNFE} and induction as shown.
\paragraph*{$(i)$} Since $\Gc_1 = \Gc = \flower(X,Y)$, we show that if $\Gc_i$ is of the form $\flower(X\setminus A_i', Y\setminus A_i')$  where $A_i\subseteq Y$. For convienience, we will call $X\setminus A_i'$ as $X_i$ and $Y\setminus A_i'$ as $Y_i$. We will show that $\Gc_{i+1}$ is of the form $\flower(X\setminus A_{i+1}', Y\setminus A_{i+1}')$ by showing that in fact it is $\flower(X_i\setminus A_{i+1}', Y_i\setminus A_{i+1}')$  for some $A_{i+1}'\subseteq Y$. First notice that $\Gc_i' = \Gc_i\setminus Attr_{\Even}^{\Gc_i}(D_i)$, where $D_i$ is some subset of Even vertices. From Lemma~\ref{prop:attractorNFEEven}, we have for $Z = X\setminus Y$,
$$\Gc_i\setminus Attr_{\Even}^{\Gc_i}(D_i) = \flower(X_i, Y_i')\setminus \flower(Z\cup Pre^*_{\Gc_i,\Even}(D_i),Pre^*_{\Gc_i,\Even}(D_i))$$
Since $Z = X\setminus Y = X_i\setminus Y_i$, we have 
$$\Gc_i' = \Gc_i' = \flower(X_i, Y_i\setminus Pre^*_{\Gc_i,\Even}(D_i))$$

The $U_i$ computed by performing \solveOdd on $\Gc_i'$ must be of the form $ \flower(X_i, Z_i)$ for $Z_i\subseteq Y_i$ by induction and the attractor to $U_i$,  must be of the form $\flower(X_i, W_i)$ from Proposition~\ref{prop:attractorNFEEven}. Hence $$\Gc_{i+1} = \flower(X_i,Y_i)\setminus \flower(X_i, W_i) = \flower(X_i\setminus W_i, X_i\setminus W_i).$$

\paragraph*{$(ii)$} We will show that if $\Gc_i$ is of the form $\flower(X,Y_i)$, then $\Gc_{i+1}$ is of the form $\flower(X, Y_{i+1})$ for $Y_{i+1}\subseteq Y_i$. 
In each iteration $i$, the Odd attractor to $D_i$ in $\Gc_i$ is of the form $\flower(X,A_i)$. This shows that $\Gc_i'$, which is obtained by removing the Odd attractor $\flower(X,A_i)$ from $\Gc_i$ is of the form $\flower(X_i\setminus A_i,Y_i\setminus A_i)$. 
Running \solveOdd on $\Gc_i$ gives $U_i$ of the form $\flower(X_i \setminus W_i, Y_i\setminus W_i)$ by induction, and an Even attractor to the set $\flower(X_i \setminus W_i, Y_i\setminus W_i)$  would be of the flowery subgame $\flower(X_i \setminus W_i', Y_i\setminus W_i')$ for some $W_i'\subseteq W_i$. So, $\Gc_{i+1}$, which is obtained from removing this Even attractor from $\Gc_i$ would be obtained as follows
\[
    \Gc_{i+1} \:= \Gc_i\setminus \flower(X_i \setminus W_i', Y_i\setminus W_i') = \flower(X, Y_i\setminus W_i').\qedhere
\]
\end{proof}

\section{Symbolic algorithm}

In this appendix we describe how the number of symbolic set variables
in the symbolic implementation of the universal algorithm
can be further reduced from~$O(d)$ to~$O(\log d)$, leading to Theorem~\ref{thm:symbolic}.

\begin{proofappendix}{thm:symbolic}{\symbolic}
  We use letters $G$, $D$, $G'$, and~$U$ to denote the sets $V^{\Gc_i}$,
  $D_i$, $V^{\Gc'_i}$, and $U_i$ for some $i$-th iteration of any of the
  recursive calls of the universal algorithm. 
  Observe that we do not need to keep the symbolic variables that store
  the sets $D$, $G'$, and~$U$ on the stack of recursive calls because on
  any return from a recursive call, their values are not needed to
  proceed.  
  How can we store the sets denoted by all the symbolic set
  variables~$G$ on the stack using only $O(\log d)$ symbolic set
  variables, while the height of the stack may be as large as~$d$?   

  Firstly, we argue that we can symbolically represent a sequence 
  $\seq{G_{d-1}, \dots, G_i}$ of set variables that would normally  
  occur on the stack of recursive calls of the universal algorithm, by
  another sequence $\seq{H_{d-1}, \dots, H_0}$, in which the sets form a
  partition of the set of vertices in the parity game. 
  Indeed, a sequence $\seq{G_d, \dots, G_i}$ on the stack of recursive
  calls at any time forms a descending chain w.r.t.\ inclusion, and
  $G_d$ is the set of all vertices, so it suffices to consider the
  sequence 
  $\seq{G_d \setminus G_{d-1}, \dots, G_{i+1} \setminus G_i, G_i,
    \emptyset, \dots, \emptyset}$. 

  Secondly, we argue that the above family of $d$ mutually disjoint sets
  can be succinctly represented and maintained using~$O(\log d)$ set
  variables.
  W.l.o.g., assume that $d$ is a power of~$2$.
  For every $k = 1, 2, \dots, \lg d$, and for every 
  $i = 1, 2, \dots, d$, let $\mathrm{bit}_k(i)$ be the $k$-th digit
  in the binary representation of~$i$
  (and zero if there are less than $k$ digits). 
  We now define the following sequence of sets 
  $\seq{S_1, S_2, \dots, S_{\lg d}}$ that provides a succinct
  representation of the sequence $\seq{H_{d-1}, \dots, H_0}$. 
  For every $k = 1, 2, \dots, \lg d$, we set:
  \[
    S_k \: = \: 
    \bigcup \left\{H_i \: : \: 
      0 \leq i \leq d-1 \text{ and } \mathrm{bit}_k(i)=1\right\}\,.
  \]
  By sets $\seq{H_{d-1}, \dots, H_0}$ forming a partition of the set of
  all vertices, it follows that for every $i = 0, 1, \dots, d-1$, we
  have: 
  \[
    H_i \: = \:
    \bigcap \left\{S_k \: : \: 
      1 \leq k \leq \lg d \text{ and } \mathrm{bit}_k(i)=1\right\} 
    \cap
    \bigcap \left\{\overline{S_k} \: : \: 
      1 \leq k \leq \lg d \text{ and } \mathrm{bit}_k(i)=0\right\}\,,
  \]
  where $\overline{X}$ is the complement of set~$X$. 

  What remains to be shown is that the operations on the sequence of
  sets $\seq{G_{d-1}, \dots, G_i}$ that reflect changes on the stack of
  recursive calls of the universal algorithm can indeed be implemented
  using small numbers of symbolic set operations on the succinct
  representation $\seq{S_1, \dots, S_{\lg d}}$ of the sequence
  $\seq{H_{d-1}, \dots, H_0}$. 
  We note that there are two types of changes to the sequence
  $\seq{G_{d-1}, \dots, G_i}$ that the universal algorithm makes:
  \begin{enumerate}[label=(\alph*)]
  \item
    all components are as before, except for $G_i$ that is replaced  
    by~$G_i \setminus B$, for some set~$B \subseteq G_i$;

  \item
    all components are as before, except that a new entry $G_{i-1}$ is 
    added equal to $G_i \setminus B$, for some set~$B \subseteq G_i$.  
  \end{enumerate}
  The corresponding changes to the sequence 
  $\seq{H_{d-1}, \dots, H_0}$ are then: 
  \begin{enumerate}[label=(\alph*)]
  \item
    \label{item:slice-U}
    all components are as before, except that set $H_{i+1}$ is replaced
    by $H_{i+1} \cup B$, and set $H_i$ is replaced by~$H_i \setminus B$;

  \item
    \label{item:recursive-call}
    all components are as before, except that set $H_i$ is replaced
    by~$B$, and set $H_{i-1}$ is replaced by~$H_i \setminus B$. 
  \end{enumerate}
  To implement the update of type~\ref{item:slice-U}, it
  suffices to perform the following update to the succinct
  representation: 
  \[
    S'_k \: = \:
    \begin{cases}
      S_k  &  \text{if $\mathrm{bit}_k(i+1) = \mathrm{bit}_k(i)$}, \\
      S_k \cup B  &  
        \text{if $\mathrm{bit}_k(i+1)=1$ and $\mathrm{bit}_k(i)=0$}, \\ 
      S_k \setminus B  &  
        \text{if $\mathrm{bit}_k(i+1)=0$ and $\mathrm{bit}_k(i)=1$}. 
    \end{cases}
  \]
  and to to implement the update of
  type~\ref{item:recursive-call}, it suffices to perform the following: 
  \[
    S'_k \: = \:
    \begin{cases}
      S_k  &  \text{if $\mathrm{bit}_k(i) = \mathrm{bit}_k(i-1)$}, \\
      S_k \setminus (H_i \setminus B)  &  
        \text{if $\mathrm{bit}_k(i)=1$ and $\mathrm{bit}_k(i-1)=0$}, \\ 
      S_k \cup (H_i \setminus B)  &  
        \text{if $\mathrm{bit}_k(i)=0$ and $\mathrm{bit}_k(i-1)=1$}.
    \end{cases}\qedhere
  \]
\end{proofappendix}

\end{document}